\documentclass{lmcs}
\pdfoutput=1

\usepackage{lastpage}
\lmcsdoi{22}{1}{16}
\lmcsheading{}{\pageref{LastPage}}{}{}%
{Dec.~25,~2024}{Feb.~27,~2026}{}


\usepackage[utf8]{inputenc}
\usepackage{xspace}
\usepackage{relsize}
\usepackage{amsmath}
\usepackage{amssymb}
\usepackage{hyperref}

\usepackage{tikz}
\usetikzlibrary{shapes.geometric} 
\usetikzlibrary{shapes.arrows} 
\usetikzlibrary{patterns}

\newcommand{\ie}{\emph{i.e.}\xspace}
\newcommand{\G}{\ensuremath{\mathcal{G}}}
\newcommand{\Graph}[3]{\ensuremath{{\G}_{#1,#2,#3}}}
\newcommand{\Gsdp}{\ensuremath{\Graph{\Sigma}{\Delta}{\pi}}}

\newcommand{\D}{\ensuremath{\mathcal{D}}}
\newcommand{\Disk}[3]{\ensuremath{{\D}_{#1,#2,#3}}}
\newcommand{\Dsdp}{\ensuremath{\Disk{\Sigma}{\Delta}{\pi}}}

\newcommand{\GDisk}[3]{\ensuremath{\mathcal{D}_{#1,#2,#3}}}
\newcommand{\GDsdp}{\ensuremath{\GDisk{\Sigma}{\Delta}{\pi}}}
\newcommand{\lan}{\Phi}
\newcommand{\Csdp}{\ensuremath{\mathbf{G}_{\Sigma,\Delta,\pi}}}
\newcommand{\id}{\ensuremath{\mathrm{id}}}
\newcommand{\Sym}{\mathrm{Sym}}
\newcommand{\U}{\ensuremath{\mathrm{U}}}

\newcommand{\CDsdp}{\ensuremath{\mathbf{D}_{\Sigma,\Delta,\pi}}}

\DeclareMathOperator{\Conj}{Conj}
\DeclareMathOperator{\Proj}{Proj}
\DeclareMathOperator{\Colim}{Colim}
\DeclareMathOperator{\card}{card}

\newcommand{\emptygraph}{\varnothing}

\theoremstyle{plain} 

\begin{document}

\title{Causal Graph Dynamics and Kan Extensions}
\thanks{This publication was made possible through the support of the ID \#62312 grant from the John Templeton Foundation, as part of the \href{https://www.templeton.org/grant/the-quantum-information-structure-of-spacetime-qiss-second-phase}{‘The Quantum Information Structure of Spacetime’ Project (QISS) }. The opinions expressed in this project/publication are those of the author(s) and do not necessarily reflect the views of the John Templeton Foundation.
We also acknowledge the support of INRIA Saclay through the funding of two years of ``delegation'' of the first author to work on these topics.}

\author[L.~Maignan]{Luidnel Maignan\lmcsorcid{0009-0009-5297-5022}}
\author[A.~Spicher]{Antoine Spicher\lmcsorcid{0009-0003-8471-4413}}

\address{Univ Paris Est Creteil, LACL, 94000, Creteil, France}
\email{luidnel.maignan@u-pec.fr, antoine.spicher@u-pec.fr}

\begin{abstract}
On the one hand, the formalism of Global Transformations comes with the claim of capturing any transformation of space that is local, synchronous and deterministic.
The claim has been proven for different classes of models such as mesh refinements from computer graphics, Lindenmayer systems from morphogenesis modeling, and cellular automata from biological, physical and parallel computation modeling.
The Global Transformation formalism achieves this by using category theory for its genericity, and more precisely the notion of Kan extension to determine the global behaviors based on the local ones.
On the other side, Causal Graph Dynamics describe the transformation of port graphs in a synchronous and deterministic way.
In this paper, we show the precise sense in which the claim of Global Transformations holds for them as well.
This is done by showing different ways in which they can be expressed as Kan extensions, each of them highlighting different features of Causal Graph Dynamics.
Along the way, this work uncovers the interesting class of Monotonic Causal Graph Dynamics and their universality among General Causal Graph Dynamics.
\end{abstract}

\maketitle

\section*{Introduction}

\paragraph{Initial Motivation.}
This work started as an effort to understand the framework of Causal Graph Dynamics (CGDs) from the point of view of Global Transformations (GTs), both frameworks expanding on Cellular Automata (CAs) with the similar goal of handling dynamical spaces, but with two different answers.
Indeed, we have on the one hand CGDs that have been introduced in 2012 in~\cite{DBLP:conf/icalp/ArrighiD12} as a way to describe synchronous and local evolutions of labeled port graphs whose structures also evolve.
Since then, the framework has evolved to incorporate many considerations such as stochasticity, reversibility~\cite{DBLP:journals/nc/ArrighiMP20} and quantumness~\cite{arrighi2017quantum}.
On the other hand, we have GTs that have been proposed in 2015 in~\cite{DBLP:conf/gg/MaignanS15} as a way to describe synchronous local evolution of any spatial structure whose structure also evolves.
This genericity over arbitrary kinds of space is obtained using the language of category theory.
It should therefore be the case that CGDs are special cases of GTs in which the spatial structure happens to be labeled port graphs.
So the initial motivation is to make this relationship precise.

\paragraph{Initial Plan.}
Initially, we expected this to be a very straightforward work.
First, because of the technical features of CGDs (recalled in the next section), it is appropriate to study them in the same way that CAs had been studied in~\cite{DBLP:journals/nc/FernandezMS23}.
By this, we mean that although GTs use the language of category theory, the categorical considerations all simplify into considerations from order theory when an absolute position system is used, as is the case in CAs and in CGDs.
Secondly, CGDs come directly with an order from the notion of sub-graph used implicitly throughout.
The initial plan was to unfold the formalism on this basis to make sure that everything is as straightforward as expected, and then proceed to the next step: quotienting absolute positions out to only keep relative ones, thus actually using categorical features and not only order-theoretic ones, as done in~\cite{DBLP:journals/nc/FernandezMS23} for CAs.

Section~\ref{sec:cgd} recalls how CGDs are defined and work.
Section~\ref{sec:poset-kan} recalls how Kan extensions and GTs are defined and simplifies in the particular case of order theory.
Section~\ref{sec:hope} makes clear the strong relation between the two concepts and how this relation leads to the initial ambition.

\paragraph{Actual Plan.}
It turns out that the initial ambition falls short, but the precise way in which it does reveals something interesting about CGDs.
The expected relationship is in fact partial and only allows to accommodate Kan extensions with CGDs which happen to be monotonic.
This led us to change our initial plan to investigate the role played by monotonic CGDs within the framework.
Doing so, we uncover the universality of monotonic CGDs among general CGDs, thus implying the initially wanted result: all CGDs are GTs.
We then take into account renaming-invariance in a categorical way: by considering port-graphs that differ only by their names to be isomorphic, which amounts to a quotienting of the simple order-like category into a smaller category.

Sections~\ref{sec:falls-short} and~\ref{sec:mono-cgd-are-kan} make clear the relationship between GTs and monotonic CGDs, while Section~\ref{sec:universal-mono-cgd} shows the universality of monotonic CGDs by proposing an encoding transforming any CGD into a monotonic one.
Section~\ref{sec:cgd_lan} is devoted to deal with renaming-invariance in a categorical way.
Finally, Section~\ref{sec:concl} discusses these results.

\section{Preliminaries}
\label{sec:preliminaries}

\subsection{Notations}

Sections~\ref{sec:preliminaries},~\ref{sec:cgd_post_lan} and~\ref{sec:universal-mono-cgd} mostly use common notations from set theory.
The set operations symbols, especially set inclusion $\subseteq$ and union $\cup$, are heavily overloaded, but the context always allows to recover the right semantics.
Two first overloads concern the inclusion and the union of partial functions, which are to be understood as the inclusion and union of the graphs of the functions respectively, \ie, their sets of input-output pairs.
A partial function $f$ from a set $A$ to a set $B$ is indicated as $f: A \rightharpoonup B$, meaning that $f$ is defined for some elements of $A$.
The restriction of a function $f: A \rightharpoonup B$ to a subset $A' \subseteq A$ is denoted $f \restriction A': A' \rightharpoonup B$.

Section~\ref{sec:cgd_lan} uses elements of category theory.
The reader is expected to be familiar with the notions of category, functor, natural transformation, colimit and comma category~\cite{maclane2013categories}.

\subsection{Causal Graph Dynamics}
\label{sec:cgd}

\newcommand{\V}{\ensuremath{\mathcal{V}}}

Our work strongly relies on the objects defined in~\cite{DBLP:journals/iandc/ArrighiD13}.
We recall here the definitions required to understand the rest of the paper.
Particularly, \cite{DBLP:journals/iandc/ArrighiD13} establishes the equivalence of the so-called \emph{causal dynamics} and \emph{localisable functions}.
For the present work, we focus on the latter.
We generally stick to the original notations to ease the comparison with the previous work.
However, some slight differences exist but are locally justified.

We consider an uncountable infinite set $\V$ of symbols for naming vertices.

\begin{defi}[Labeled Graph with Ports]
\label{def:graph}
Let $\Sigma$ and $\Delta$ be two sets, and $\pi$ a finite set.
A \emph{graph $G$ with states in $\Sigma$ and $\Delta$, and ports in $\pi$} is the data of:
\begin{itemize}

\item a countable subset $V(G) \subset \V$ whose elements are the \emph{vertices} of $G$,

\item a set $E(G)$ of non-intersecting two-element subsets of $V(G) \times \pi$, whose elements are the  \emph{edges} of $G$,

\item a partial function $\sigma(G): V(G) \rightharpoonup \Sigma$ labeling vertices of $G$ with states,

\item a partial function $\delta(G): E(G) \rightharpoonup \Delta$ labeling edges of $G$ with states.
\end{itemize}
The set of graphs with states in $\Sigma$ and $\Delta$, and ports $\pi$ is written $\Gsdp$.

A \emph{pointed graph} $(G,v)$ is a graph $G \in \Gsdp$ with a selected vertex $v \in V(G)$ called \emph{pointer}.
\end{defi}

As in the original paper, we use the notation $u\!:\!i$ for elements  $(u,i) \in V(G) \times \pi$.
The fact that edges are non-intersecting means that each $u:i \in V(G) \times \pi$ appears in at most one edge. The degree of any vertex, \ie the number of edges containing it with any port, is therefore at most $\card(\pi)$, the number of ports.
Let us note right now that the particular elements of $\V$ used to build $V(G)$ in a graph $G$ should ultimately be irrelevant, and only the structure and the labels should matter, as made precise in condition (1) of Definition~\ref{def:graph_local_rule}.
This is the reason of the definition of the following object.

\begin{defi}[Renaming]
\label{def:graph_iso}
\label{def:graph_renaming}
A \emph{renaming\footnote{Renamings are originally called isomorphisms. To avoid ambiguity, we reserve the use of the latter for categorical purpose.}} is the data of a bijection $R: \mathcal{V} \rightarrow \mathcal{V}$.
Its action on vertices is straightforwardly extended to any edge by
$
R(\{u:i,v:j\}) := \{R(u):i,R(v):j\}),
$
to any graph $G \in \mathcal{G}$ by
\begin{align*}
V(R(G)) &:= \{ R(v) \mid v \in V(G) \}, &
E(R(G)) &:= \{ R(e) \mid e \in E(G) \}, \\
\sigma(R(G)) &:= \sigma(G) \circ R^{-1}, &
\delta(R(G)) &:= \delta(G) \circ R^{-1},
\end{align*}
and to any pointed graph $(G,v)$ by $R(G,v) := (R(G),R(v))$.
\end{defi}

Operations are defined to manipulate graphs in a set-like fashion.

\begin{defi}[Consistency, Union, Intersection]
\label{def:graph_cons}
\label{def:graph_union}
\label{def:graph_inter}
Two graphs $G$ and $H$ are \emph{consistent} when
\begin{itemize}
\item $E(G) \cup E(H)$ is a set of non-intersecting two-element sets;
\item $\sigma(G)$ and $\sigma(H)$ agree where they are both defined;
\item $\delta(G)$ and $\delta(H)$ agree where they are both defined.
\end{itemize}
In this case, the \emph{union} $G \cup H$ of $G$ and $H$ is defined by
\begin{align*}
V(G \cup H) &:= V(G) \cup V(H), &
E(G \cup H) &:= E(G) \cup E(H), \\
\sigma(G \cup H) &:= \sigma(G) \cup \sigma(H), &
\delta(G \cup H) &:= \delta(G) \cup \delta(H).
\end{align*}
The \emph{intersection} $G \cap H$ of $G$ and $H$ is always defined and given by
\begin{align*}
V(G \cap H) &:= V(G) \cap V(H), &
E(G \cap H) &:= E(G) \cap E(H), \\
\sigma(G \cap H) &:= \sigma(G) \cap \sigma(H), &
\delta(G \cap H) &:= \delta(G) \cap \delta(H).
\end{align*}
The \emph{empty graph} with no vertex is designated $\emptygraph$.
\end{defi}
In this definition, the unions of functions seen as relations (that result is a non-functional relation in general) are here guaranteed to be functional by the consistency condition.
The intersection of (partial) functions simply gives a (partial) function which is defined only on inputs on which both functions agree.

To describe the evolution of such graphs in the CGD framework, we first need to make precise the notion of locality, which is captured by how a graph is pruned to a local view: a disk.
We consider the usual notion of distance between vertices a graph $G$ as the length of shortest paths between them.
A path $p$ of length $n$ is a sequence of $n+1$ vertices such that any two consecutive vertices $p_k$ and $p_{k+1}$ are connected by an edge via some ports.
The distance function is expressed formally as
$$d(u,v) := \min\;\{\, n \in \mathbb{N} \mid \exists p \in V(G)^{n+1}, \forall k \in \{1,\ldots,n\}, \exists i, j \in \pi, \{\, p_k\!:\!i \,,\, p_{k+1}\!:\!j \,\} \in E(G) \,\}.$$
We denote by $B_G(c,r)$ the set of vertices at distance at most $r$ from $c$ in the graph $G$, \ie, $B_G(c,r) = \{ u \in V(G) \mid d(c,u) \le r\}$.
\begin{defi}[Disk]
\label{def:graph_disk}
Let $G$ be a graph, $c \in V(G)$ a vertex, and $r$ a non-negative integer.
The \emph{disk of radius $r$ and center $c$} is the pointed graph $G^r_c = (H,c)$ with $H$ given by
\begin{align*}
V(H) &:= B_G(c,r+1), &
\sigma(H) &:= \sigma(G) \restriction B_G(c,r), \\
E(H) &:= \{ \{u\!:\!i,v\!:\!j\} \in E(G) \mid \{u,v\} \cap B_G(c,r) \neq \emptyset \}, &
\delta(H) &:= \delta(G) \restriction E(H).
\end{align*}
We denote by $\GDsdp^r$ the set $\{\, G^r_c \mid G \in \Gsdp, c \in V(G) \,\}$ of all disks of radius $r$, and by $\GDsdp$ the set of all disks of any radius.
When the disk notation is used with a set $C$ of vertices as subscript, we mean
\begin{equation}
\label{eq:graph_restriction}
G^r_C := \bigcup_{c \in C} G^r_c,
\end{equation}
which is well defined as any pair of disks taken in a common graph are always consistent.
\end{defi}
Notice that $V(H)$ also contains vertices at distance $r+1$.
The reason is that a disk of radius $r$ is defined so as to contain all edges incident to any vertex at distance at most $r$.
Some of these edges are between a vertex at distance exactly $r$ and a vertex at distance $r+1$.
So those $(r+1)$-vertices need to belong to $V(H)$ to express those edges.
Notice also that the $\sigma$ function of the disk does keep the label of $(r+1)$-vertices.

The CGD dynamics relies on a local evolution describing how local views generates local outputs consistently.

\begin{defi}[Local Rule]
\label{def:graph_local_rule}
A \emph{local rule of radius $r$} is a function $f: \GDsdp^r \rightarrow \Gsdp$ such that
\begin{enumerate}

\item for any renaming $R$, there is another renaming $R'$, called the \emph{conjugate} of $R$, with $f \circ R = R' \circ f$,

\item for any family $\{(H_i,v_i)\}_i \subset \GDsdp^r$, $\bigcap_{i} H_i = \emptygraph \implies \bigcap_{i} f((H_i,v_i)) = \emptygraph$,

\item there exists $b$ such that for all $D \in \GDsdp^r$, $\card(V(f(D))) \le b$,

\item for any $G \in \Gsdp$, $u, v \in V(G)$, $f(G^r_v)$ and $f(G^r_u)$ are consistent.
\end{enumerate}
\end{defi}

In the second condition, note that a set of graphs have the empty graph as intersection iff their sets of vertices are disjoint.
In the original work, functions respecting the two first conditions, the third condition, and the fourth condition are called respectively dynamics, bounded functions, and consistent functions.
A local rule is therefore a bounded consistent dynamics.

In~\cite{DBLP:journals/iandc/ArrighiD13}, a \emph{causal graph dynamics} is an endofunction on $\Gsdp$ which respects a kind of continuity. 
The main result of~\cite{DBLP:journals/iandc/ArrighiD13} is the non-trivial proof that they coincide with \emph{localisable functions}, \ie, functions arising from local rules by union as below.
We rely on this result in the following definition since we use the formal definition of localisable functions to define causal graph dynamics.

\begin{defi}[Causal Graph Dynamics (CGD)]
\label{def:graph_cgd}
A function $F: \Gsdp \rightarrow \Gsdp$ is a \emph{causal graph dynamics} (CGD) if there exists a radius $r$ and a local rule $f$ of radius $r$ such that
\begin{equation}\label{eq:cgd}
F(G) = \bigcup_{v \in V(G)} f(G^r_v).
\end{equation}
\end{defi}
In Equation~\ref{eq:cgd}, the union is well defined thanks to the property (4) of local rules.
An important property of CGDs is their renaming-invariance which is inherited from the properties of local rules.
\begin{prop}[Renaming Invariance]
\label{prop:cgd-rename-invariance}
Let $F: \Gsdp \rightarrow \Gsdp$ be a CGD.
For any renaming $R$, $F$ admits at least a conjugate renaming $R'$, \ie, $F \circ R = R' \circ F$.
Moreover, for any local rule $f$ of $F$, $\Conj_f(R) \subseteq \Conj_F(R)$ where $\Conj_F(R)$ (resp. $\Conj_f(R)$) is the set of conjugates of a renaming $R$ for $F$ (resp. $f$).
\end{prop}

\begin{proof}
Let $F$ be a CGD arising from a local rule $f$, and $R$ a renaming.
Consider a conjugate $R' \in \Conj_f(R)$.
We have
\begin{align*}
F(R(G)) & = \bigcup_{v \in V(R(G))} f(R(G)^r_v) \\
& = \bigcup_{v \in R(V(G))} f(R(G^r_{R^{-1}(v)})) \\
& = \bigcup_{R^{-1}(v) \in V(G)} R'(f(G^r_{R^{-1}(v)})) \\
& = R'\left(\bigcup_{v \in V(G)} f(G^r_v)\right) \\
& = R'(F(G))
\end{align*}
So, we have $\Conj_f(R) \subseteq \Conj_F(R)$.
Moreover, since $\Conj_f(R)$ is not empty by property (1) of Definition~\ref{def:graph_local_rule}, so is $\Conj_F(R)$.
\end{proof}

\subsection{Global Transformations \& Kan Extensions}
\label{sec:poset-kan}

In category theory~\cite{maclane2013categories}, Kan extensions are a construction allowing to extend a functor along another one in a universal way.
In the first part of this article, we restrict ourselves to the case of pointwise left Kan extensions involving only categories which are posets.
In this case, their definition simplifies as follows.
(The general case is given in Section~\ref{sec:cgd_lan}.)
\begin{defi}[Pointwise Left Kan Extension for Posets]
Given three posets $A$, $B$ and $C$, and two monotonic functions $i : A \to B$, $f : A \to C$, the function $\lan: B \to C$ given by 
\begin{equation}
\label{eq:pointwise}
\lan(b) = \sup \,\{\, f(a) \in C \mid a \in A \text{ s.t. } i(a) \preceq b \,\}
\end{equation}
is called the \emph{pointwise left Kan extension of $f$ along $i$} when it is well defined (some suprema might not exist), and in which case it is necessarily monotonic.
\end{defi}

Global Transformations (GTs) make use of left Kan extensions to tackle the question of the synchronous deterministic local transformation of \emph{arbitrary} kinds of spaces.
It is a categorical framework, but in the restricted case of posets, it works as follows.
While $B$ and $C$ capture as posets the local-to-global relationship between the spatial elements to be handled (inputs and outputs respectively), $A$ specifies a poset of local transformation rules.
The monotonic functions $i$ and $f$ give respectively the left-hand-side and right-hand-side of the rules in $A$.
Glancing at Eq.~\eqref{eq:pointwise}, the transformation mechanism works as follows.
Consider an input spatial object $b \in B$ to be transformed.
The associated output $\lan(b)$ is obtained by gathering (thanks to the supremum in $C$) all the right-hand-side $f(a)$ of the rules $a \in A$ with a left-hand-side $i(a)$ occurring in $b$.
The occurrence relationship is captured by the respective partial orders: $i(a) \preceq b$ in $B$ for the left-hand-side, and $f(a) \preceq \lan(b)$ in $C$ for the right-hand-side.

The monotonicity of $\lan$ is the formal expression of a major property of a GT: if an input $b$ is a subpart of an input $b'$ (\ie, $b \preceq b'$ in $B$), the output $\lan(b)$ has to occur as a subpart of the output $\lan(b')$ (\ie, $\lan(b) \preceq \lan(b')$ in $C$).
This property gives to the orders of $B$ and $C$ a particular semantics for GTs which will play an important role in the present work.
\begin{rem}
\label{rem:gt-poset-as-informational-poset}
Elements of $B$ are understood as information about the input.
So, when $b \preceq b'$, $b'$ provides a richer information than $b$ about the input that $\lan$ uses to produce output $\lan(b')$, itself richer than output $\lan(b)$.
However, $\lan$ cannot deduce the falsety of a property about the input from the fact this property is not included in $b$; otherwise the output $\lan(b)$ might be incompatible with $\lan(b')$.
\end{rem}
At the categorical level, the whole GT formalism relies on the key ingredient that the collection $A$ of rules is also a category.
Morphisms in $A$ are called \emph{rule inclusions}.
They guide the construction of the output and allow overlapping rules to be applied all together avoiding the well-known issue of concurrent rule application~\cite{DBLP:conf/gg/MaignanS15}.

In cases where $\lan$ captures the evolution function of a (discrete time) dynamical system (so particularly for the present work where we want to compare $\lan$ to a CGD), we consider the input and output categories/posets $B$ and $C$ to be the same category/poset, making $\lan$ an endo-functor/function.

\section{Unifying Causal Graph Dynamics and Kan Extensions}
\label{sec:cgd_post_lan}

The starting point of our study is that Eq.~\eqref{eq:cgd} in the definition of CGDs has almost the same form as Eq.~\eqref{eq:pointwise}.
Indeed, if we take Eq.~\eqref{eq:pointwise} and set $A = \GDsdp^r$, $B = C = \Gsdp$, the function $i$ to be the pointer dropping function from discs to graphs that drops their centers (\ie, $i((H,c)) = H$) and the function $f$ to be the local rule from discs to graphs, we obtain an equation for $\lan$ of the form
\begin{equation}
\label{eq:pointwise2}
\lan(G) = \sup \,\{\, f(D) \in \Gsdp \mid D \in \GDsdp^r \text{ s.t. } i(D) \preceq G \,\}
\end{equation}
which is close to Eq.~\eqref{eq:cgd} rewritten
$$
F(G) = \bigcup \,\{\, f(D) \in \Gsdp \mid D = G^r_v,\, v \in V(G) \,\}.
$$
This brings many questions.
What is the partial order involved in Eq.~\eqref{eq:pointwise2}?
Is the union of  Eq.~\eqref{eq:cgd} given by the suprema of this partial order?
Is it the case that $i(D) \preceq G$ implies $D = G^r_v$ for some $v$ in this order?
Are $f$ and $F$ of Eq.~\eqref{eq:cgd} monotonic functions for this order?
We tackle these questions in the following sections.

\subsection{The Underlying Partial Order}
\label{sec:hope}

Considering the two first questions, there is a partial order which is forced on us.
Indeed, we need this partial order to imply that suprema are unions of graphs.
But partial orders can be defined from their binary suprema since $A \preceq B \iff \sup \,\{A,B\} = B$.
Let us give explicitly the partial order, since it is very natural, and prove afterward that it is the one given by the previous procedure.

\begin{defi}[Subgraph and Subdisk]
\label{def:graph_inc}
Given two graphs $G$ and $H$, $G$ is a \emph{subgraph} of $H$, denoted $G \subseteq H$, when
$$
V(G) \subseteq V(H) \ \wedge\ 
E(G) \subseteq E(H) \ \wedge\ 
\sigma(G) \subseteq \sigma(H) \ \wedge\ 
\delta(G) \subseteq \delta(H).
$$
This defines a partial order $- \subseteq -$ on $\Gsdp$ called the \emph{subgraph order}.
The subgraph order is extended to pointed graphs by
$$
(G,v) \subseteq (H,u) \ :\!\iff\  G \subseteq H \wedge v = u.
$$
This partial order restricted to $\Dsdp$ is called the \emph{subdisk order}.
\end{defi}
In this definition, the relational condition $\sigma(G) \subseteq \sigma(H)$, where these two functions are taken as sets of input-output pairs, means that $\sigma(G)(v)$ is either undefined or equal to $\sigma(H)(v)$, for any vertex $v \in V(G)$.
The same holds for the condition $\delta(G) \subseteq \delta(H)$.

Let us now state that the subgraph order has the correct relation with unions of graphs.
It similarly encodes consistency and intersections of pairs of graphs.

\begin{prop}
Two graphs $G$ and $H$ are consistent precisely when they admit an upper bound in $(\Gsdp, \subseteq)$.
In this case, the union of $G$ and $H$ is exactly their supremum (least upper bound) in $(\Gsdp, \subseteq)$.
The intersection of $G$ and $H$ is exactly their infimum (greatest lower bound) in $(\Gsdp, \subseteq)$.
\end{prop}
\begin{proof}
Admitting an upper bound in $\Gsdp$ means that there is $K \in \Gsdp$ such that $G \subseteq K$ and $H \subseteq K$.
Since the subgraph order is defined componentwise, we consider the union of $G$ and $H$ as quadruplets of sets:
$$
(V(G) \cup V(H), E(G) \cup E(H), \sigma(G) \cup \sigma(H), \delta(G) \cup \delta(H)),
$$
which always exists and is the least upper bound in the poset of quadruplets of sets with the natural order.
For this object to be a graph, it is enough to check $E(G) \cup E(H)$ is a non-intersecting two-element set, and that $\sigma(G)$ and $\sigma(H)$ (resp. $\delta(G)$ and $\delta(H)$) coincide on their common domain.
This is the case when $G$ and $H$ admit an upper bound as required in the definition.
Conversely, when the condition holds, the union is itself an upper bound.

For intersection, consider 
$$
(V(G) \cap V(H), E(G) \cap E(H), \sigma(G) \cap \sigma(H), \delta(G) \cap \delta(H))
$$
which is clearly the greatest lower bound and is always a graph.
\end{proof}

\begin{rem}
\label{rem:renamings-monotonic}
Note that the pointer dropping function $i : \Dsdp^r \to \Gsdp$ is monotonic.
Any renaming $R : \Gsdp \to \Gsdp$ is monotonic as well.
The underlying reason is that $R$ basically acts elementwise on the set $V(G)$ of each graph $G$, a monotonic operation for set inclusion.
\end{rem}

The two first questions being answered positively, let us instantiate Eq.~\eqref{eq:pointwise2} the subgraph order.
\begin{equation}
\label{eq:pointwise3}
\lan(G) = \bigcup \,\{\, f(D) \mid D \in \GDsdp^r \text{ s.t. } i(D) \subseteq G \,\}
\end{equation}

\subsection{Comparing Disks and Subgraphs}
\label{sec:falls-short}

\newcommand{\SetY}{I_v}

Let us embark on the third question: is it the case that $i(D) \subseteq G$ implies $D = G^r_v$ for some $v$?
Making the long story short, the answer is no.
But it is crucial to understand precisely why.
Fix some vertex $v \in V(G)$.
Clearly, in Eq.~\eqref{eq:cgd}, the only considered disk centered on $v$ is $G^r_v$.
Let us determine now what are exactly the disks centered on $v$ involved in Eq.~\eqref{eq:pointwise3}, that is, the set $\SetY := \{ (H,v) \in \GDsdp^r \mid H \subseteq G \}$.
Firstly, $G^r_v$ is one of them of course.

\begin{lem}
\label{prop:disk_in_Y}
For any vertex $v\in V(G)$, $G^r_v \in \SetY$.
\end{lem}

\begin{proof}
In Definition~\ref{def:graph_disk}, the graph component $H$ of the pointed graph $G^r_v$ is explicitly defined by taking a subset for each of the four components of $G$, as required in Definition~\ref{def:graph_inc} of subgraph.
\end{proof}

The concern is that $G^r_v$ is generally not the only one disk in $\SetY$ as expected by Eq.~\eqref{eq:cgd}.
However, $G^r_v$ is the maximal one in the following sense.

\begin{lem}
\label{prop:disk_top}
For any vertex $v\in V(G)$, consider the disk $G^r_v = (H,v)$.
Then for any disk $(H',v) \in \SetY$, we have $H' \subseteq H$.
\end{lem}

\begin{proof}
By Definition~\ref{def:graph_inc} of the subgraph order, we need to prove four inclusions.
For the inclusion of vertices, consider an arbitrary vertex $w \in V(H')$.
By definition of disks of radius $r$, there is a path in $H'$ from $v$ to $w$ of length at most $r+1$.
But since $H' \subseteq G$, we have $w \in V(G)$ and this path itself is also in $G$.
So $w \in V(G)$ respects the defining property of the set $V(G^r_v)$ and therefore belongs to it.
The three other inclusions ($E(H') \subseteq E(H)$, $\sigma(H') \subseteq \sigma(H)$ and $\delta(H') \subseteq \delta(H)$) are proved similarly, by using the definition of disks, then the fact $H' \subseteq G$, and finally the definition of $G^r_v$.
\end{proof}

In some sense, the converse of the previous proposition holds: it is roughly enough to be smaller than $G^r_v$ to be a disk of $\SetY$, as characterized by the following two facts.

\begin{lem}
\label{prop:disk_conn}
The set of disks $\GDsdp$ is the set of pointed connected finite graphs.
\end{lem}

\begin{proof}
Indeed, for any disk $(H,v) \in \GDsdp$, $H$ is connected since all vertices are connected to $v$, and $H$ is finite since all vertices have at most $\card(\pi)$ neighbors, so a rough bound is $\card(\pi)^r$ with $r$ the radius of the disk.
Conversely, for any pointed connected finite graphs $(G,v)$, we have $(G,v) = G^{r}_v$ for any $r \ge \card(V(G))-1$, the length of the longest possible path.
\end{proof}

\begin{prop}
\label{prop:Iv_poset}
For any vertex $v\in V(G)$, $\SetY$ is a principal downward closed set in the poset of graphs restricted to connected finite graphs containing $v$.
\end{prop}

\begin{proof}
Indeed, consider any disk $(H,v)$ such that $H \subseteq G$.
Now, take $H'$ a connected graph containing $v$ and such that $H' \subseteq H$.
Since $H$ is finite, so is $H'$.
By Lemma~\ref{prop:disk_conn}, $(H',v)$ is also a disk.
By transitivity, $H' \subseteq H \subseteq G$.
This proves that we have a downward closed set.
This is moreover a principal one because of Lemmas~\ref{prop:disk_in_Y} and~\ref{prop:disk_top}.
\end{proof}

We now know that the union of Eq.~\eqref{eq:pointwise3} receives a bigger set of local outputs to merge than the union of Eq.~\eqref{eq:cgd}.
But we cannot conclude anything yet.
Indeed, it might be the case that all additional local outputs do not contribute anything more.
This is in particular the case if disks $D \in \SetY$ are such that $f(D) \subseteq f(G^r_v)$.
This is related to the fourth and last question.

\subsection{Monotonic and General Causal Graph Dynamics}
\label{sec:mono-cgd-are-kan}

The last remark invites us to consider the case where the CGD is monotonic.
We deal with the general case afterward.

\subsubsection{Monotonic CGDs as Kan Extensions.}

As just evoked, things seem to go well if the local rule $f$ happens to be monotonic.
All the ingredients have been already given and the proposition can be made formal straightforwardly.

\begin{prop}
\label{prop:mcgd-are-ke}
Let $F: \Gsdp \rightarrow \Gsdp$ be a CGD with local rule $f: \GDsdp^r \rightarrow \Gsdp$ of radius $r$.
If $f$ is monotonic, then $F$ is the pointwise left Kan extension of $f$ along $i: \GDsdp^r \rightarrow \Gsdp$, the pointer dropping function.
\end{prop}

\begin{proof}
The proposition is equivalent to show that $F(G) = \lan(G)$ for all $G$.
We proceed by double-inclusion.
Summarizing our journey up to here, we now know that Eq.~\eqref{eq:pointwise3} and Eq.~\eqref{eq:cgd} are similar except that the former iterates over the set of disks $I = \bigcup_{v \in V(G)} \SetY$ while the latter iterates over $J = \{ G^r_v \mid v \in V(G) \}$, with $J \subseteq I$ by Proposition~\ref{prop:disk_in_Y}.
We get $F(G) \subseteq \lan(G)$.
Moreover, by Proposition~\ref{prop:disk_top}, for any $D \in \SetY$ for any $v$, $D \subseteq G^r_v$, and by monotonicity of $f$, $f(D) \subseteq f(G^r_v)$.
So $f(D) \subseteq F(G)$ for all $D$ and $\lan(G) \subseteq F(G)$, leading to the expected equality.
\end{proof}

The class of CGDs having such a monotonic local rule is easily characterizable: they correspond to CGDs that are monotonic themselves as stated by the following proposition.

\begin{prop}
\label{prop:cgd_monotonic_iff_local_rule_monotonic}
A CGD $F$ is monotonic iff $F$ admits a monotonic local rule.
\end{prop}

\begin{proof}
If $F$ has a monotonic local rule $f$, then $F$ is a left Kan extension by Proposition~\ref{prop:mcgd-are-ke} and is monotonic as recalled in Section~\ref{sec:poset-kan}.

Conversely, suppose $F$ monotonic.
By Definition~\ref{def:graph_cgd} of CGDs, there is a local rule $f'$, not necessarily monotonic, of radius $r$ generating $F$.
Consider $f$ defined by $f((H,c)) = F(H)$ for any disk $(H,c)$.
$f$ is monotonic by monotonicity of $F$.
$f$ is a local rule, which is checked easily using that $f'$ is itself a local rule and that $f((H,c)) = \bigcup_{v \in H} f'(H^r_v)$.
We now show that $f$ generates $F$.
For any graph $G$:
$$
\bigcup_{c \in V(G)} f(G^r_c) = \bigcup_{c \in V(G)} \bigcup_{v \in G^r_c} f'(i(G^r_c)^r_v).
$$
But $f'(i(G^r_c)^r_v) \subseteq F(i(G^r_c)) \subseteq F(G)$, the last inclusion coming from the monotonicity of $F$.
So $\bigcup_{c \in V(G)} f(G^r_c) \subseteq F(G)$.
Moreover, $f'(G^r_c) = f'(i(G^r_c)^r_c)$ for any $c$, so $f'(G^r_c) \subseteq \bigcup_{c \in V(G)} f(G^r_c)$ and $F(G) \subseteq \bigcup_{c \in V(G)} f(G^r_c)$.
\end{proof}

\begin{cor}
\label{cor:monototonic-cgd-iff-lan}
A CGD is monotonic iff it is a left Kan extension. 
\end{cor}

\subsubsection{The Non-Monotonic Case.}

\newcommand{\bnode}[4]{
\node [shape=circle,minimum size=0.6cm,draw,fill=#3] at (#2) (#1) {#4};
\node[xshift=-7.pt, yshift=-5.pt] at (#1.west) {$:\!l$};
\node[xshift=7.pt, yshift=5.pt] at (#1.east) {$:\!r$};
\fill[black] (#1.east)++(-2pt,-1pt) rectangle ++(4pt,2pt);
\fill[black] (#1.west)++(-2pt,-1pt) rectangle ++(4pt,2pt);
}

\begin{figure}[t]
\centering
\begin{tikzpicture}[every loop/.style={},scale=.8, every node/.style={scale=.8}]
\begin{scope}
\bnode{B}{-4,1}{green!20}{}
\end{scope}
\node [shape=single arrow,fill=gray,minimum height=1cm] at (0,1) {};
\begin{scope}
\bnode{B}{4,1}{green!20}{$\,\blacktriangleright$}
\end{scope}

\begin{scope}
\bnode{A}{-6,0}{blue!20}{}
\bnode{B}{-4,0}{green!20}{$\,\blacktriangleright$}
\bnode{C}{-2,0}{red!20}{}
\draw (A) edge (B);
\draw (B) edge (C);
\end{scope}
\node [shape=single arrow,fill=gray,minimum height=1cm] at (0,0) {};
\begin{scope}
\bnode{A}{2,0}{blue!20}{}
\bnode{B}{4,0}{green!20}{}
\bnode{C}{6,0}{red!20}{$\,\blacktriangleright$}
\draw (A) edge (B);
\draw (B) edge (C);
\end{scope}

\begin{scope}
\bnode{A}{-6,-1}{blue!20}{}
\bnode{B}{-4,-1}{green!20}{$\,\blacktriangleright$}
\draw (A) edge (B);
\end{scope}
\node [shape=single arrow,fill=gray,minimum height=1cm] at (0,-1) {};
\begin{scope}
\bnode{A}{2,-1}{blue!20}{}
\bnode{B}{4,-1}{green!20}{$\blacktriangleleft\,$}
\draw (A) edge (B);
\end{scope}
\end{tikzpicture}
\caption{\label{fig:non-monotonic-particle}
Moving particle CGD - non-monotonic behavior.
Each row represents an example of evolution with a graph $G$ on the left and $F(G)$ on the right.
Colors correspond to vertex names.}
\end{figure}
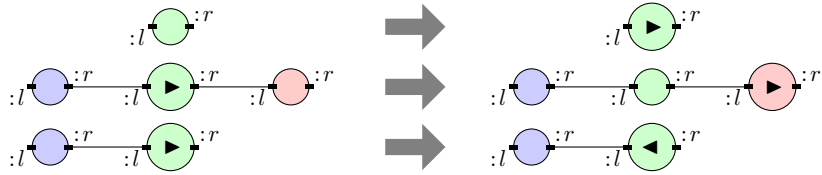

The previous result brings us close to our initial goal: encoding any CGD as a GT.
The job would be considered done only if there is no non-monotonic CGDs, or if those CGDs are degenerate cases.
However, it is clearly not the case and most of the examples in the literature are of this kind, as we can see with the following example, inspired from~\cite{DBLP:journals/nc/ArrighiMP20}.

Consider for example the modeling of a particle going left and right on a linear graph by bouncing at the extremities.
The linear structure is coded using two ports on each vertex, say $l$ and $r$ for left and right respectively, while the particle is represented with the presence of a label on some vertex, with two possible values indicating its direction.
See Figure~\ref{fig:non-monotonic-particle} for illustrations of such graphs.
The dynamics of the particle is captured by $F$ as follows.
Suppose that the particle is located at some vertex $v$ (in green in Figure~\ref{fig:non-monotonic-particle}), and wants to go to the right.
If there is an outgoing edge to the right to an unlabeled vertex $u$ (in red in Figure~\ref{fig:non-monotonic-particle}), the label representing the particle is moved from $v$ to $u$ (second row of Figure~\ref{fig:non-monotonic-particle}).
If there is no outgoing edge to the right (the port $r$ of $v$ is free), it bounces by becoming a left-going particle (third row of Figure~\ref{fig:non-monotonic-particle}).
$F$ works symmetrically for a left-going particle.

The behavior of $F$ is non-monotonic since the latter situation is a sub-graph of the former, while the particle behaviors in the two cases are clearly incompatible.
On Figure~\ref{fig:non-monotonic-particle}, the right hand sides of the second and third rows are comparable by inclusion, but the left hand sides are not.
This non-monotonicity involves a missing edge but missing labels may induce non-monotonicity as well.
Suppose for the sake of the argument that $F$ generates a right-going particle at any unlabeled isolated vertex (first row of Figure~\ref{fig:non-monotonic-particle}).
The unlabeled one-vertex graph is clearly a subgraph of any other one where the same vertex is labeled by a right-going particle and has some unlabeled right neighbor.
In the former case, the dynamics puts a label on the vertex, while it removes it on the latter case.
The new configurations are no longer comparable.
See the first and second rows of Figure~\ref{fig:non-monotonic-particle} for an illustration.

\begin{prop}
\label{prop:cgd_not_necessary_monotonic}
CGDs are not necessarily monotonic.
\end{prop}

\begin{proof}
Take $\card(\pi) = 1$, $\Sigma = \Delta = \emptyset$.
In a graph $G$, we have isolated vertices and pairwise connected vertices.
For the local rule, consider $r = 1$, so the two possible disks (modulo renaming) are the isolated vertex and a pair of connected vertices.
There is no graph $G$ such that the two disks appear together (otherwise $G$ would ask the vertex to be connected and disconnected at the same time).
So we can define $f$ such that it acts inconsistently on the two disks, since property 4 of local rules does not apply here.
The isolated vertex is included in the connected vertices in the sense of $\subseteq$, but the image by $f$, then by $F$ are not.
\end{proof}

From Corollary~\ref{cor:monototonic-cgd-iff-lan}, we conclude that all non-monotonic  CGDs are not left Kan extensions as we have developed so far, \ie, based on the subgraph relationship of Def~\ref{def:graph_inc}.
Analyzing the particle CGD in the light of Remark~\ref{rem:gt-poset-as-informational-poset} tells us why.
Indeed, the subgraph ordering is able to compare a place without any right neighbor with a place with some (left-hand-side of rows 3 and 2 in Figure~\ref{fig:non-monotonic-particle} for instance).
Following Remark~\ref{rem:gt-poset-as-informational-poset}, in the GT setting, the former situation has less information than the latter: in the former, there is no clue whether the place has a neighbor or not; the dynamics should not be able to specify any behavior for a particle at that place.
But clearly, for the corresponding CGD, both situations are totally different: the former is an extremity while the latter is not, and the dynamics specifies two different behaviors accordingly for a particle at that place.

\section{Universality of Monotonic Causal Graph Dynamics}
\label{sec:universal-mono-cgd}

\newcommand{\orig}{0}
\newcommand{\shut}{1}

We have proven that the set of all CGDs is strictly bigger than the set of monotonic ones.
However, we prove now that it is not more expressive.
By this, we mean that we can simulate any CGD by a monotonic one, \ie, monotonic CGDs are universal among general CGDs.

More precisely, given a \emph{general} CGD $F$, a monotonic simulation of $F$ consists of an encoding, call it $\omega(G)$, of each graph $G$, and a \emph{monotonic} CGD $F'$ such that whenever $F(G) = H$ on the \emph{general} side, $F'(\omega(G)) = \omega(H)$ on the \emph{monotonic} side.
Substituting $H$ in the latter equation using the former equation, we get $F'(\omega(G)) = \omega(F(G))$, the exact property of the expected simulation: for any $F$, we want some $\omega$ and $F'$ such that $F' \circ \omega = \omega \circ F$.

\subsection{Key Ideas of the Simulation}

\newcommand{\pnode}[4]{
\node [shape=circle,minimum size=0.6cm,draw,fill=#3] at (#2) (#1) {#4};
\node[xshift=-7.pt, yshift=-5.pt] at (#1.west) {$:\!l$};
\node[xshift=7.pt, yshift=5.pt] at (#1.east) {$:\!r$};
\node[xshift=-8.pt, yshift=4.pt] at (#1.north) {$:\!l'$};
\node[xshift=8.pt, yshift=-4.pt] at (#1.south) {$:\!r'$};
\fill[black] (#1.south)++(-1pt,-2pt) rectangle ++(2pt,4pt);
\fill[black] (#1.north)++(-1pt,-2pt) rectangle ++(2pt,4pt);
\fill[black] (#1.east)++(-2pt,-1pt) rectangle ++(4pt,2pt);
\fill[black] (#1.west)++(-2pt,-1pt) rectangle ++(4pt,2pt);
}

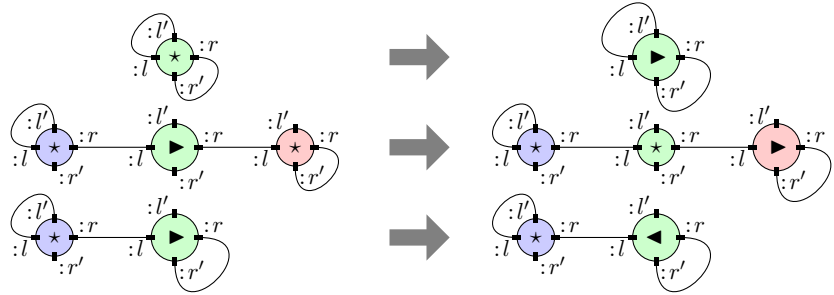
\begin{figure}[t]
\centering
\begin{tikzpicture}[every loop/.style={},scale=.8, every node/.style={scale=.8}]
\begin{scope}
\pnode{B}{-4,1.5}{green!20}{$\star$}
\draw (B) edge[out=180,in=90,loop,looseness=6] (B);
\draw (B) edge[out=0,in=270,loop,looseness=6] (B);
\end{scope}
\node [shape=single arrow,fill=gray,minimum height=1cm] at (0,1.5) {};
\begin{scope}
\pnode{B}{4,1.5}{green!20}{$\,\blacktriangleright$}
\draw (B) edge[out=180,in=90,loop,looseness=6] (B);
\draw (B) edge[out=0,in=270,loop,looseness=6] (B);
\end{scope}

\begin{scope}
\pnode{A}{-6,0}{blue!20}{$\star$}
\pnode{B}{-4,0}{green!20}{$\,\blacktriangleright$}
\pnode{C}{-2,0}{red!20}{$\star$}
\draw (A) edge[out=180,in=90,loop,looseness=6] (A);
\draw (A) edge (B);
\draw (B) edge (C);
\draw (C) edge[out=0,in=270,loop,looseness=6] (C);
\end{scope}
\node [shape=single arrow,fill=gray,minimum height=1cm] at (0,0) {};
\begin{scope}
\pnode{A}{2,0}{blue!20}{$\star$}
\pnode{B}{4,0}{green!20}{$\star$}
\pnode{C}{6,0}{red!20}{$\,\blacktriangleright$}
\draw (A) edge[out=180,in=90,loop,looseness=6] (A);
\draw (A) edge (B);
\draw (B) edge (C);
\draw (C) edge[out=0,in=270,loop,looseness=6] (C);
\end{scope}

\begin{scope}
\pnode{A}{-6,-1.5}{blue!20}{$\star$}
\pnode{B}{-4,-1.5}{green!20}{$\,\blacktriangleright$}
\draw (A) edge[out=180,in=90,loop,looseness=6] (A);
\draw (A) edge (B);
\draw (B) edge[out=0,in=270,loop,looseness=6] (B);
\end{scope}
\node [shape=single arrow,fill=gray,minimum height=1cm] at (0,-1.5) {};
\begin{scope}
\pnode{A}{2,-1.5}{blue!20}{$\star$}
\pnode{B}{4,-1.5}{green!20}{$\blacktriangleleft\,$}
\draw (A) edge[out=180,in=90,loop,looseness=6] (A);
\draw (A) edge (B);
\draw (B) edge[out=0,in=270,loop,looseness=6] (B);
\end{scope}
\end{tikzpicture}
\caption{\label{fig:monotonic-particle}
Moving particle CGD - monotonic behavior.
Compared to Figure~\ref{fig:non-monotonic-particle}, the vertices have two additional ports ($l'$ and $r'$) and unlabeled vertices are now labeled by $\star$.}
\end{figure}

In this section, we aim at introducing the key elements of the simulation informally and by means of the moving particle example.

\subsubsection{Encoding the Original Graphs: The Moving Particle Case.}

Let us design a monotonic simulation of the particle dynamics.
The original dynamics can be made monotonic by replacing the two missing edges at the extremities by easily identifiable loopback edges, making incomparable the two originally comparable situations.
For such loopback edges to exist, we need an additional port for each original port, in our case say $l'$ and $r'$ for instance.
For the case of non-monotonicity with labels, vertices where there is no particle (so originally unlabeled) are marked with a special label, say $\star$.
Figure~\ref{fig:monotonic-particle} depicts the same evolutions as Figure~\ref{fig:non-monotonic-particle} after those transformations.

This will be the exact role of the encoding function $\omega$: the key idea to design a monotonic simulation is to make uncomparable the initially comparable situations. 
Any missing entity (edge or label) composing an original graph $G$ needs to be replaced by a special entity (loopback edge or $\star$ respectively) in $\omega(G)$ indicating it was originally missing.
At the end of the day, for any $G \subsetneq H$, we get $\omega(G)$ and $\omega(H)$ no longer comparable.

\subsubsection{An Extended Set of Graphs.}

All that remains is to design $F'$ such that $F' \circ \omega = \omega \circ F$.
It is important to note however that the set of graphs targeted by $\omega$, which is the domain of $F'$, is by design much larger than the original one.
Indeed, each vertex has now a doubled number of ports and the additional label $\star$ is available.
So to be completely done with the task, we need $F'$ to be able to work not only with graphs generated by the encoding, but also with all the other graphs.

Let us classify the various cases.
Firstly, notice that the counterpart $\omega(G)$ of any graph $G$ is ``total'' in the following sense: all vertices have labels, all original ports have edges, and all edges have labels.
However, there also exist partial graphs with free ports and unlabeled vertices not targeted by $\omega$.
Secondly, $\omega(G)$ uses the additional ports strictly for the encoding of missing edges with loopback edges.
But, there are also graphs making arbitrary use of those ports and which are not ``coherent'' with respect to the encoding.
Figure~\ref{fig:graph-classification} illustrates the three identified classes: the middle graph is a coherent subgraph of the total graph on the left, and of an incoherent graph on the right.
Notice that incoherent edges can always be dropped away to get the largest coherent subgraph of any graph.
This is the case on the figure.

\begin{figure}[t]
\centering
\begin{tikzpicture}[every loop/.style={},scale=.8, every node/.style={scale=.8}]
\begin{scope}
\pnode{A}{-6,0}{blue!20}{$\star$}
\pnode{B}{-4,0}{green!20}{$\,\blacktriangleright$}
\draw (A) edge[out=180,in=90,loop,looseness=6] (A);
\draw (A) edge (B);
\draw (B) edge[out=0,in=270,loop,looseness=6] (B);
\end{scope}
\begin{scope}
\pnode{A}{-1,0}{blue!20}{$\star$}
\pnode{B}{1,0}{green!20}{}
\draw (A) edge[out=180,in=90,loop,looseness=6] (A);
\end{scope}
\begin{scope}
\pnode{A}{4,0}{blue!20}{$\star$}
\pnode{B}{6,0}{green!20}{}
\draw (A) edge[out=180,in=90,loop,looseness=6] (A);
\draw (A) edge[out=0,in=90, looseness=1.3] (B);
\end{scope}
\end{tikzpicture}
\caption{\label{fig:graph-classification}
Different classes of graphs: from left to right, a total graph, a coherent partial graph, an incoherent graph.}
\end{figure}
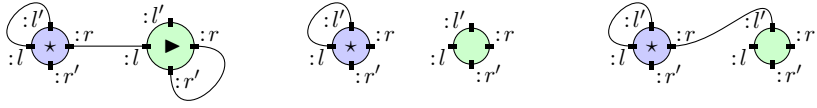

\subsubsection{An Extended Set of Disks.}

In order to design a local rule $f'$ of the monotonic CGD $F'$, we need to handle disks after encoding.
The encoding function $\omega$ targeting a bigger set of graphs, this also holds for the encoding function of disks, and the behavior of $f'$ needs to handle this bigger set.

Let us first identify the encoded counterparts of the original disks.
They are called ``total'' disks and correspond to disks of some $\omega(G)$.
For those disks, the injectivity of $\omega$ allows $f'$ to simply retrieve the corresponding disk of the original graph $G$ and invoke the original local rule $f$ on it.

An arbitrary disk $D$ may not be ``total'' and exhibit some partiality (free ports, unlabeled vertices or edges).
Since $f'$ is required to be monotonic, it needs to output a subgraph of the original local rule output.
More precisely, we need to have $f'(D) \subseteq f'(D')$ for any total disk $D'$ with $D \subseteq D'$.
The easiest way to do that is by outputting the empty graph: $f'(D) = \emptygraph$. (This solution corresponds to the so-called \emph{coarse extension} proposed for CAs in~\cite{DBLP:journals/nc/FernandezMS23}; \emph{finer extensions} are also considered there.)

Last, but not least, an arbitrary disk $D$ may make an incoherent use of the additional ports with respect to $\omega$.
In that case, all incoherent information may be ignored by considering the largest coherent subdisk $D' \subseteq D$.
The behavior of $f'$ on $D$ is then aligned with its behavior on $D'$: $f'(D) = f'(D')$.

\subsubsection{A Larger Radius.}

The last parameter of $f'$ to tune is its radius.
It will of course depend on the radius of the local rule of $F$.
But it is worth noting that we are actually trying to build more information locally than the original local rule was trying to.
Indeed, given a vertex of the original global output, many local rule applications may contribute concurrently to its definition.
All of these contributions are consistent with each other of course, but it is possible for some local outputs to indicate some features of that vertex (label, edges) while others do not.
If even one of them puts such a label for instance, then there is a label in the global output.
It is only if none of them put a label that the global output will not have any label on it.
The same holds similarly for ports: a port is free in the global output if it is so on all local outputs.
Because the monotonic counterpart needs to specify locally if none of the local rule applications puts such a feature to a vertex, the radius of the monotonic local rule needs to be big enough to include all those local rule input disks.
It turns out that the required radius for $f'$ is $r' = 3 r + 2$, if $r$ is the radius of the original local rule $f$, as explained in Figure~\ref{fig:bigger-radius}.

\begin{figure}[!t]
\centering
\begin{tikzpicture}
\node (U) at (0,0) {$c$};
\node (V) at (2.5,0) {$v$};
\node (M) at (1.25,0) {};
\draw [|-|] (0,2) to node[above]{$r+1$} (1.25,2);
\draw [|-|] (1.25,2) to node[above]{$r+1$} (2.5,2);
\draw [|-|] (2.5,2) to node[above]{$r$} (3.5,2);
\draw  (U) circle (1.0);
\draw [, dashed] (U) circle (1.5);
\draw  (V) circle (1.0);
\draw [, dashed] (V) circle (1.5);
\draw [fill=black] (M) circle (0.125);
\end{tikzpicture}
\caption{\label{fig:bigger-radius}
We need to determine for each missing entity attached to a vertex or edge of $f(G^r_c)$, whether all other $f(G^r_v)$ agree to consider this entity as missing.
Property 2 of local rules (Definition~\ref{def:graph_local_rule}) tells us that only disks $G^r_v$ that intersect $G^r_c$ need to be checked.
The furthest such $v$ are at distance $2r+2$.
From there, we need to ask for radius $r$ (so radius $r'=3r+2$ from $c$) to include the entirety of $G^r_v$ (including its border at $r+1$ from $v$, and therefore $3r+3$ at most from $c$).}
\end{figure}
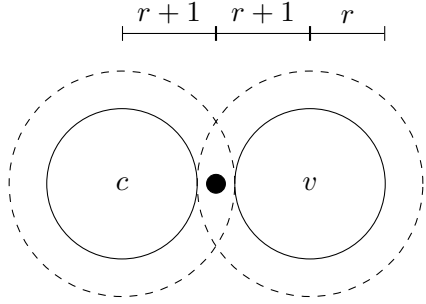

\subsection{Formal Definition of the Simulation}

Now that we have all the components of the solution, let us make them precise.

\subsubsection{The Encoding Function $\omega$.}

The encoding function aims at embedding the graphs of $\Gsdp$ into graphs where ports are doubled and labels are extended with an extra symbol.
Set $\Sigma' := \Sigma \uplus \{ \star \}$, $\Delta' := \Delta \uplus \{ \star \}$, and $\pi' = \pi \times \{ \orig, \shut \}$.
Ports $(a,\orig) \in \pi'$ are considered to be semantically the same as their counterpart $a \in \pi$ while ports $(a,\shut)$ are there for loopback edges.
For simplicity, let us define the short-hands $\G$ and $\G'$ for $\Gsdp$ and $\Graph{\Sigma'}{\Delta'}{\pi'}$ respectively for the remainder of this section.
Since we need to complete partial functions, we introduce the following notation:
for any partial function $f: X \rightharpoonup Y$ and total function $z: X \to Z$, we define the total function $f!z : X \to Y\cup Z$ by
$$
(f!z)(x) =
\begin{cases}
f(x) & \text{if } f(x) \text{ is defined,} \\
z(x) & \text{otherwise.} 
\end{cases}
$$
For simplicity, we write $f!\star$ for $f!(x \mapsto \star)$ where $x$ runs over $X$.
This allows us to define the encoding $\omega(G) \in \G'$ of any $G \in \G$ as follows.
\begin{align*}
V(\omega(G)) & := V(G) \\
\sigma(\omega(G)) & := \sigma(G)!\star \\
E(\omega(G)) & := \{\, \{ u\!:\!(i,\orig), \mathrm{trg}^\star_{G}(u\!:\!i) \} \mid u\!:\!i \in V(G)\!:\!\pi \,\} \\
\delta(\omega(G)) & (\{u\!:\!(i,a),v\!:\!(j,b)\})  := 
\begin{cases}
(\delta(G)!\star)(\{u\!:\!\orig,v\!:\!j\}) & \text{if } b = \orig \\
\star & \text{if } u\!:\!i = v\!:\!j
\end{cases}
\end{align*}
The definition of $E(\omega(G))$ is written to make clear that all original ports are indeed occupied by an edge.
The total function $\mathrm{trg}^\star_{G} : V(G)\times\pi \to V(G)\times\pi'$ is defined based on the partial function $\mathrm{trg}_{G} : V(G)\times\pi \rightharpoonup V(G)\times\pi$ as follows.
\begin{align*}
\mathrm{trg}_G&(u\!:\!i) := v\!:\!j \text{ iff } \{u\!:\!i,v\!:\!j\} \in E(G) \\
\mathrm{trg}^\star_{G} & := (t_\orig \circ\mathrm{trg}_G) ! t_\shut \text{ where } t_m(u\!:\!i) := u\!:\!(i,m)
\end{align*}
The definition of $\delta(\omega(G))$ deals with original edges for the first case, and with loopback edges for the second, which are the only two possibilities with respect to $\mathrm{trg}^\star_{G}$.

\subsubsection{The Disk Encoding Function $\omega_r$.}

Let us now define the shorthands $\D_r$ and $\D'_r$ for $\Dsdp^r$ and $\Disk{\Sigma'}{\Delta'}{\pi'}^r$ for any radius $r$.
We put the radius as a subscript to make readable the four combinations $\D_r$, $\D'_r$, $\D_{r'}$, and $\D'_{r'}$.
The function $\omega_r: \D_r \to \D'_r$ aims at encoding original disks as total disks.
It is defined as
\begin{equation}
\omega_r((H,v)) = \omega(H)^r_v. 
\end{equation}
The reason we apply $(-)^r_v : \G' \to \D'_r$ on the raw result of $\omega$ is because of Definition~\ref{def:graph_disk} of disks.
Indeed, a disk of radius $r$ contains by definition unlabeled vertices at distance $r+1$.
The behavior of $\omega$ (adding loopback edges on each unused port and labels on unlabeled vertices and edges) on vertices at distance $r+1$ makes the image not a disk in $\D'_r$.
This is again because disks in $\D'_r$ should have no label on $(r+1)$-vertices and no edge (nor loopback edge) between $(r+1)$-vertices.
More conceptually, $\omega_r$ is the unique function having the following commutation property, which expresses that $\omega_r$ is the disk counterpart of $\omega$.
\begin{prop}
\label{prop:commutation_omega_disk}
For any graph $G \in \G$, $\omega_r(G^r_v) = \omega(G)^r_v$.
\end{prop}

\begin{proof}
Take $G^r_v = (H,v)$.
After inlining the definition of $\omega_r$, we are left to show $\omega(H)^r_v = \omega(G)^r_v$.
For any entity (vertices, labels, edges) at distance at most $r$ of $v$, $\omega(H)$ and $\omega(G)$ coincide.
At distance $r+1$, things may differ for the special label $\star$ and loopback edges, and beyond distance $r+1$, $\omega(H)$ has no entities at all.
However, all these differences are exactly what is removed by $(-)^r_v$.
So, we get the expected equality.
\end{proof}

\subsubsection{The Coherent Subgraph Function $\coh$.}

\newcommand{\coh}{\mathrm{coh}}

As discussed informally, some graphs in $\G'$ may be ill-formed from the perspective of the encoding by using doubled ports arbitrarily.
We define the function $\coh : \G' \to \G'$ that removes the incoherencies.
\begin{align*}
V(\coh(G')) & := V(G') \hspace{1cm} \\
\sigma(\coh(G')) & := \sigma(G') \hspace{1cm} \\
E(\coh(G')) & := \{ \{u\!:\!(i,\orig), v\!:\!(j,b)\} \in E(G') \mid\; \\
  & \phantom{ := \{ \{} b = 1 \Rightarrow (v\!:\!j = u\!:\!i \;\wedge\; \delta(G')(\{u\!:\!(i,\orig), u\!:\!(i,1)\}) = \star ) \} \\
\delta(\coh(G')) & := \delta(G') \restriction E(\coh(G'))
\end{align*}
Since $\coh$ only removes incoherent edges, and since they are stable by inclusion, we trivially have that
\begin{lem}
\label{prop:coh-monotonic}
$\coh$ is monotonic.
\end{lem}

\subsubsection{The Monotonic Local Rule.}

\newcommand{\Image}{\ensuremath{\mathrm{Im}}}

The monotonic local rule $f'$ is defined in two stages.
The first stage is to take a disk in $\D'_{r'}$ and transform it into something that the original local rule can work with.
To do this, we remove any incoherencies with the function $\coh$, and if the result is total, we restrict it to the correct radius and retrieve its preimage by $\omega_r$ (which is possible because $\omega_r$ is injective).
The original local function can be called on this counterpart.
As discussed before, if the coherent sub-disk is not total, we simply output an empty graph.
The signature of this function is $\phi : \D'_{r'} \to \G$.
\begin{equation}
\phi((H',c)) := 
\left\{
\begin{array}{ll}
\emptygraph & \text{if } \coh(H') \notin \Image(\omega_{r'}) \\
f(\omega_r^{-1}(\coh(H')^r_c)) & \text{otherwise}
\end{array}
\right.
\end{equation}
Turning this result in $\G$ into its counterpart in $\G'$ is more complicated than simply invoking $\omega$ since one has to check that missing entities are really missing, as discussed in Figure~\ref{fig:bigger-radius}.
This is the purpose of the second stage leading to the definition of the monotonic local rule $f' : \D'_{r'} \to \G'$ which, thanks to its radius $r' = 3r+2$, is able to inspect all $r$-disks at distance $2r + 2$ from $c$.
\begin{equation}
f'((H',c)) = \omega\left({\textstyle \bigcup_{v \in H'^{2r+1}_c}} \phi((H',v))\right)^0_{\phi((H',c))}
\end{equation}
In this equation, the union of all local results that can contribute to entities attached to $\phi((H',c))$ is built.
In this way, they are all given a chance to say if what was missing in $\phi((H',c))$ is actually missing.
Finally $\omega$ is applied and its result restricted to the sole vertices of $\phi((H',c))$ and their edges, using the notation $(-)^0_X$ of Eq.~\eqref{eq:graph_restriction}.
Those vertices and edges are the only ones for which all the possibly contributing disks have been inspected.

\begin{prop}
\label{prop:fprime_monotonic}
$f'$ is a monotonic local rule.
\end{prop}

\begin{proof}
We first show that $f'$ is a local rule, that is, it satisfies the four properties of Definition~\ref{def:graph_local_rule}.

We need to show that for any renaming $R$, there is a conjugate $R'$ such that $f' \circ R = R' \circ f'$.
We show that the conjugate $R'$ of $R$ for $f$ works.
The result is obvious for $\coh(H') \notin \Image(\omega_{r'})$.
In the other case, take $E' := \coh(H')$:
\begin{align*}
f'(R((H',c)))
& = f'((R(H'), R(c))) \\
& = \omega\left(\bigcup_{v \in R(H')^{2r+1}_{R(c)}} f(\omega_r^{-1}(R(E')^r_v))\right)_{f(\omega_r^{-1}(R(E')^r_{R(c)}))} \\
& = \omega\left(\bigcup_{v \in R(H')^{2r+1}_{R(c)}} R'(f(\omega_r^{-1}(E'^r_{R'(v)})))\right)_{R'(f(\omega_r^{-1}(E'^r_c)))} \\
& = R'\left[\omega\left(\bigcup_{v \in H'^{2r+1}_{c}} f(\omega_r^{-1}(E'^r_{v}))\right)\right]_{R'(f(\omega_r^{-1}(E'^r_c)))} \\
& = R'\left[\omega\left(\bigcup_{v \in H'^{2r+1}_{c}} f(\omega_r^{-1}(E'^r_{v}))\right)_{f(\omega_r^{-1}(E'^r_c))}\right] \\
& = R'(f((H',c))).
\end{align*}

Consider a family of disks $\{(H'_i,c_i)\}_{i \in I} \subset \D'_{r'}$ such that $\bigcap_i H'_i = \emptygraph$.
We need to show that $\bigcap_i f'((H'_i,c_i)) = \emptygraph$.
We suppose that for all $i$, $E'_i := \coh(H'_i) \in \Image(\omega_{r'})$ (otherwise the result is trivial).
First notice that for any family $\{ v_i \in (H_i')^{2r+1}_{c_i} \}$, we have $\bigcap_i \omega^{-1}_r((E'_i)^r_{v_i}) = \emptygraph$, so $\bigcap_i f(\omega^{-1}_r((E'_i)^r_{v_i})) = \emptygraph$.
Suppose now some vertex $u$ in $\bigcap_i f'((H'_i,c_i))$.
So, for each $i$ there is some $v_i$ such that $u \in f(\omega^{-1}_r((E_i')^r_{v_i}))$ which is impossible.

We need a bound $b$ such that for all $(H',c) \in \D'_{r'}$, $\card(V(f'((H',c)))) \le b$.
We consider the bound $b$ given by $f$ and show that it works.
Indeed the last step of the computation of $f'$ is precisely a restriction to the vertices of $f(\omega^{-1}_r(\coh(H')^r_c))$.
So $\card(V(f'((H',c)))) \le \card(V(f(\omega^{-1}_r(\coh(H')^r_c)))) \le b$.

Consider $G' \in \G'$ and $c,d \in V(G')$.
We need $f'(G'^{r'}_c)$ and $f'(G'^{r'}_d)$ to be consistent.
Once again, we suppose that $\coh(G'^{r'}_c)\in \Image(\omega_{r'})$ and $\coh(G'^{r'}_d)\in \Image(\omega_{r'})$ (otherwise the result is trivial).
In order to use the consistency property of $f$, we build $G = \omega_r^{-1}(\coh(G'^{r'}_c)) \cup \omega_r^{-1}(\coh(G'^{r'}_d))$.
So for any pair of $r$-disks $(D_1,D_2)$ of $G$, $f(D_1)$ and $f(D_2)$ are consistent.
In other words, all the original $r$-disks involved in $f'(G'^{r'}_c)$ and $f'(G'^{r'}_d)$ are consistent with each other.
It is particularly true for $\phi(G'^{r'}_c)$ and $\phi(G'^{r'}_d)$.
Edges of $f'(G'^{r'}_c)$ and $f'(G'^{r'}_d)$ are consistent because either they come from $f(G^r_c)$ and $f(G^r_d)$ which are consistent, or they are loopback edges added by $\omega$.
Labels are also consistent for the same kind of reason, making $f'(G'^{r'}_c)$ and $f'(G'^{r'}_d)$ consistent.

We finally show that $f'$ is monotonic.
Take two disks $D'_1 = (H'_1, c)$ and $D'_2 = (H'_2, c)$ such that $D'_1 \subseteq D'_2$.
As $f'$ is ultimately defined by cases, let us consider first the case where $\coh(H'_1) \notin \Image(\omega_{r'})$.
In this case $f'(D'_1) = \emptygraph$ so we necessarily have $f'(D'_1) \subseteq f'(D'_2)$.
We are left with the case $\coh(H'_1) \in \Image(\omega_{r'})$, meaning that $\coh(H'_1)$ is a total disk.
Since $\coh$ is monotonic by Lemma~\ref{prop:coh-monotonic}, $\coh(H'_1) \subseteq \coh(H'_2)$.
And as a total disk, nothing but incoherences can be added to $\coh(H'_1)$.
So $\coh(H'_1) = \coh(H'_2)$ and therefore $f'(D'_1) = f'(D'_2)$.
So the order is preserved in all cases, and $f'$ is monotonic.
\end{proof}

\subsubsection{The Monotonic Simulation.}

We finally have the wanted CGD $F'$ of local rule $f'$:
$$
F'(G') := \bigcup_{v \in V(G')} f'({G'}^{r'}_v).
$$
It remains then to show that $F'$ is indeed a monotonic simulation of $F$, which is achieved with the two next propositions.

\begin{prop}
$F'$ is monotonic.
\end{prop}

\begin{proof}
Since $f'$ is monotonic by Proposition~\ref{prop:fprime_monotonic}, $F'$ is monotonic as well by Proposition~\ref{prop:cgd_monotonic_iff_local_rule_monotonic}.
\end{proof}

\begin{prop}
$F'$ simulates $F$ via the $\omega$ encoding, \ie, $F' \circ \omega = \omega \circ F$.
\end{prop}

\begin{proof}
Suppose $G' = \omega(G)$ for some $G \in \G$.
Since all involved disks in $F'(\omega(G))$ are total and thanks to Proposition~\ref{prop:commutation_omega_disk}, the expression of $F'(\omega(G))$ simplifies drastically to
\begin{equation}
\label{eq:prop-simulation-Fprime}
F'(\omega(G)) = \bigcup_{c \in V(G)} \omega\left({\textstyle \bigcup_{v \in G^{2r+1}_c}} f(G^r_v) \right)^0_{f(G^r_c)}.
\end{equation}
Clearly, $F'(\omega(G))$ does not exhibit any incoherencies and $\omega(F(G))$ is total, so it is enough to show that $\omega(F(G)) \subseteq F'(\omega(G))$ to have the equality.
This can be done entity by entity.
Take a vertex $u \in \omega(F(G))$.
It comes from some $f(G^r_c)$.
So it belongs to the inner union of Eq.~\eqref{eq:prop-simulation-Fprime}.
It is preserved by $\omega$ then by $(-)^0_{f(G^r_c)}$, so it belongs to $F'(\omega(G))$.
Consider now its label $\ell := \delta(\omega(F(G)))(u)$.
If $\ell \ne \star$, there is some $v \in G^{2r+1}_c$ with $\ell = \delta(f(G^r_v))(u)$.
So the inner union of Eq.~\eqref{eq:prop-simulation-Fprime} labels $u$ by $\ell$.
The label is preserved by $\omega$ then by $(-)^0_{f(G^r_c)}$ since $u \in f(G^r_c)$.
So $\delta(F'(\omega(G)))(u) = \ell$ as well.
If $\ell = \star$, this means that none of the $f(G^r_v)$ put a label on $u$, so $u$ is unlabeled in the inner union of Eq.~\eqref{eq:prop-simulation-Fprime}.
By definition, $\omega$ completes the labeling by $\star$ at $u$, and $(-)^0_{f(G^r_c)}$ preserves this label.
So $\delta(F'(\omega(G)))(u) = \star$ as well.
The proof continues similarly for edges and their labels, with an additional care for dealing with loopback edges.
\end{proof}

\section{Renaming-Invariance Categorically}
\label{sec:cgd_lan}

In Sections~\ref{sec:cgd_post_lan} and~\ref{sec:universal-mono-cgd}, we successfully establish the wanted connection between the fact that CGDs are particular cases of synchronous local dynamical systems with absolute positioning and the fact that they can be presented as posetal GTs, \ie, posetal left Kan extensions.
We now want to incorporate the renaming-invariance property of CGDs.
This property forbids CGDs to depend on the absolute positioning.
While it is treated above as a property on top of the monotonicity of the dynamics, we now show how it can be integrated in the very algebraic structure of the graphs.
This is done by turning the poset of graphs into a category of graphs where renamed graphs are considered isomorphic.
This new formulation permits, in particular, to make explicit the finite nature of the local rule by expressing directly that there are finitely many disks, but that each of them can possibly be found in many places in a given graph.

Since we consider categories instead of posets, the goal is now to show that CDGs are categorical left Kan extensions.
Let us recall the definition of the later and fix some notations.
\begin{defi}[Pointwise Left Kan Extension]
\label{def:general_pointwise_lan}
Given three categories $\mathbf{A}$, $\mathbf{B}$ and $\mathbf{C}$, and two functors $i: \mathbf{A} \to \mathbf{B}$, $f: \mathbf{A} \to \mathbf{C}$, a \emph{pointwise left Kan extension of $f$ along $i$} is any functor $F:\mathbf{B} \to \mathbf{C}$ such that
\begin{equation}
\label{eq:gen_lan}
F \cong \Colim (f \circ \Proj_{i/-})
\end{equation}
when the right hand side is well defined (some colimits might not exist).
\end{defi}
Here, $\Colim D$ designates the colimit of the functor $D$ and is the counterpart of the supremum of Eq.~\eqref{eq:pointwise}.
Secondly, the notation $i/b$ stands for the comma category of the functor $i$ over an object $b \in \mathbf{B}$ and $i/{-}: \mathbf{B} \to \mathbf{Cat}$ is the functor associated with this construction.
It is the counterpart of the (partial ordered) set of $\{a \in A \mid i(a) \subseteq b\}$ of Eq.~\eqref{eq:pointwise}.
However, objects (resp. morphisms) of $i/b$ are pairs whose first components are objects (resp. morphisms) of $\mathbf{A}$, and second components is a witness of relation with $b$ through $i$.
The functor $\Proj_{i/b}: i/b \to \mathbf{A}$ thus extracts the first component of those pairs.

For the rest of the section, we fix once and for all a \textbf{monotonic} CGD $F$ with \textbf{monotonic} local rule $f$ of radius $r$.
Our strategy is then as follows.
\begin{enumerate}

\item Extending the poset $(\Gsdp,\subseteq)$ to a category $\Csdp$ where graphs related by renaming are isomorphic.
This gives more morphisms to work with on top of the original poset morphisms.

\item Extending the monotonic function $F: \Gsdp \to \Gsdp$ to a functor $\widetilde{F}: \Csdp \to \Csdp$, which means specifying how the additional morphisms are dealt with through renaming-invariance.

\item Extending the poset of disk $\Dsdp^r$ to a category $\CDsdp^r$ to relate disks via renamings.

\item Extending the monotonic local rule function $f: \Dsdp^r \to \Gsdp$ (resp. the pointer dropping function $i: \Dsdp^r \to \Gsdp$) to a local rule functor $\widetilde{f}: \CDsdp^r \to \Csdp$ (resp. a pointer dropping functor $\widetilde{i}: \CDsdp^r \to \Csdp$).

\item Showing that $\widetilde{F} \cong \Colim \widetilde{f} \circ \Proj_{\widetilde{i}/{-}}$.
\end{enumerate}

\subsection{The Category of Graphs with Renamings}

Our first task consists in extending the previously used poset $(\Gsdp, \subseteq)$ into a category allowing to compare graphs by renaming.
Adding such renamings to the poset calls for a closure of the collection of \emph{graph morphisms} (inclusions and renamings) by composition.
This means that a morphism is in fact a (finite) composition of inclusions and renamings.
One can convince oneself that such a composition encodes the same information as a single renaming followed by an inclusion, leading to the following definition.

\begin{defi}[(Global) Subgraph Isomorphisms]
\label{def:graph_morphism}
Let $G$ and $H$ be graphs of $\Gsdp$.
A \emph{(global) subgraph isomorphism} $m: G \to H$ is the data of a renaming denoted $|m|: \V \to \V$ such that
$
|m|(G) \subseteq H
$.
\end{defi}

\begin{defi}
\label{def:graph_cat}
We define the category $\Csdp$ whose objects are graphs of $\Gsdp$ and morphisms are subgraph isomorphisms.
The identity $\id_G: G \to G$ is given by the identity renaming (\ie, $|\id_G|:=\id_\V$) and the composition by composition of the underlying renamings (\ie, $|n \circ m| := |n| \circ |m|$).
\end{defi}

\begin{prop}
\label{prop:graph_cat}
$\Csdp$ is indeed a category.
\end{prop}

\begin{proof}
Let $m: G \to H$ and $n: H \to K$ be two subgraph isomorphisms.
We check that $|n| \circ |m|$ defines a subgraph isomorphism from $G$ to $K$.
We have $|m|(G) \subseteq H$ and $|n|(H) \subseteq K$, so $(|n| \circ |m|)(G) = |n|(|m|(G)) \subseteq |n|(H) \subseteq K$ as required, where we use Remark~\ref{rem:renamings-monotonic} for the first inclusion.
We easily check that the identity renaming $\id_\V$ defines a subgraph isomorphism for any graph $G$ ($|\id_G|(G) = \id_\V(G) = G \subseteq G$), which is neutral for the composition of renamings.
Finally, associativity of the morphism composition is inherited from the associativity of renaming composition.
\end{proof}

In the following, we say morphism for subgraph isomorphism.

\begin{rem}
\label{rem:adjoint_emptyfunc_pipe}
The notation $|{-}|$ defines in fact a faithful functor from $\Csdp$ to $\Sym_\V$, the symmetric group of $\V$ seen as a category with a single object. 
Interestingly, the empty graph $\emptygraph$ is the unique graph accepting all renamings as endomorphisms.
This provides us with a left adjoint $\emptygraph_{-}$ to $|{-}|$ which maps each renaming $R \in \Sym_\V$ to the corresponding endomorphism $\emptygraph_R: \emptygraph \to \emptygraph$, \ie, $|\emptygraph_R| = R$.
\end{rem}

\begin{prop}
\label{prop:forms_of_iso}
Isomorphisms in $\Csdp$ are of the form $m: G \to |m|(G)$, and their inverses are $m^{-1}: |m|(G) \to G$ such that $|m^{-1}| = |m|^{-1}$.
\end{prop}
\begin{proof}
We have $|m^{-1}| = |m|^{-1}$ by functoriality of $|{-}|$ (Remark~\ref{rem:adjoint_emptyfunc_pipe}).
Suppose now that $m: G \to H$ an isomorphism and its inverse $m^{-1}: H \to G$.
Let us prove that $H = |m|(G)$ by double-inclusion.
For a first direction, by Definition~\ref{def:graph_morphism}, the morphism $m: G \to H$ tells us that $|m|(G) \subseteq H$.
For the other direction, the morphism $m^{-1} : H \to G$ tells us that $|m^{-1}|(H) \subseteq G$ which rewrites into
$|m|^{-1}(H) \subseteq G$ then the wanted $H \subseteq |m|(G)$ by Remark~\ref{rem:renamings-monotonic}.
\end{proof}
In fact, each renaming $R$ induces for each graph $G$ an isomorphism, written $G_R: G \to R(G)$, such that $|G_R| = R$.
Notice that this notation complies with the definition of functor $\emptygraph_{-}$ of Remark~\ref{rem:adjoint_emptyfunc_pipe}.
However $G_{-}$ is not a functor when $G \ne \emptygraph$.

We end by formally observing how $\Csdp$ embeds $(\Gsdp,\subseteq)$ as we expected.

\begin{defi}[Embedding Functor]
\label{def:embedding-functor}
The \emph{embedding functor} $\U: (\Gsdp,\subseteq) \to \Csdp$ is the functor acting as the identity on objects and where $\U(G \subseteq H): G \to H$ is such that $|\U(G \subseteq H)| := \id_\V$.
\end{defi}

\begin{prop}
$\U$ is indeed an embedding functor.
\end{prop}

\begin{proof}
It is trivial to show that $\U$ is a functor. 
It is also trivially injective on objects.
Finally, $\U$ is faithful since its domain is posetal.
\end{proof}
As expected, any morphism $m:G \to H$ can be decomposed into a renaming followed by an inclusion: $m = \U(|m|(G) \subseteq H) \circ (G_{|m|}: G \to |m|(G))$.

\subsection{CGDs as Endofunctors}

Now that we have extended the poset $\Gsdp$ to a category $\Csdp$, our next objective is to promote the monotonic CGD $F$ on $\Gsdp$ to an endofunctor $\widetilde{F}$ on $\Csdp$.
Particularly, we want $\widetilde{F}(G) = F(G)$ for all $G$, and $\widetilde{F}$ to reflect the monotonicity of $F$ (\ie, to map inclusions to inclusions).
Those properties are captured by the following commutation.
\begin{equation}
\label{eq:ftilde_commut}
\U \circ F = \widetilde{F} \circ \U
\end{equation}
Let us investigate the behavior of such a $\widetilde{F}$ on morphisms $m: G \to H$.
We first show that the underlying renaming of $\widetilde{F}(m)$ does not depend on $G$ and $H$.

\begin{prop}
\label{prop:image_R_independent}
Consider some $\widetilde{F}$ obeying to Eq.~\eqref{eq:ftilde_commut}.
For any $m: G \to H$, the renamings $|\widetilde{F}(m)|$ and $|\widetilde{F}(\emptygraph_{|m|})|$ are equal.
\end{prop}

\begin{proof}
For any $m: G \to H$, since $\emptygraph \subseteq G$ and $\emptygraph \subseteq H$, we observe the commutation
\begin{equation}
\label{eq:observed_comm}
m \circ \U(\emptygraph \subseteq G) = \U(\emptygraph \subseteq H) \circ \emptygraph_{|m|}
\end{equation}
where $\emptygraph_{|m|}: \emptygraph \to \emptygraph$ is given by Remark~\ref{rem:adjoint_emptyfunc_pipe}.
We get 
\begin{align*}
|\widetilde{F}(m)| & = |\widetilde{F}(m)| \circ \id_\V = |\widetilde{F}(m)| \circ |\U(F(\emptygraph 
\subseteq G))|\\
& = |\widetilde{F}(m)| \circ |\widetilde{F}(\U(\emptygraph \subseteq G))| & \text{by Eq.~\eqref{eq:ftilde_commut}}\\
& = |\widetilde{F}(m \circ \U(\emptygraph \subseteq G))| = |\widetilde{F}(\U(\emptygraph \subseteq H) \circ \emptygraph_{|m|})| & \text{by Eq.~\eqref{eq:observed_comm}}\\
& = |\widetilde{F}(\U(\emptygraph \subseteq H))| \circ |\widetilde{F}(\emptygraph_{|m|})|\\
& = |\U(F(\emptygraph \subseteq H))| \circ |\widetilde{F}(\emptygraph_{|m|})| & \text{by Eq.~\eqref{eq:ftilde_commut}}\\
& = \id_\V \circ |\widetilde{F}(\emptygraph_{|m|})| = |\widetilde{F}(\emptygraph_{|m|})|
\end{align*}
as required.
\end{proof}

This proposition tells us that the entire behavior of $\widetilde{F}$ on morphisms is simply determined by a function over renamings.
We denote this function $\overline{F}: \Sym_\V \to \Sym_\V$ which is defined by $\overline{F}(R) := |\widetilde{F}(\emptygraph_R)|$.
Notice the commutation $\overline{F}(|m|) = |\widetilde{F}(m)|$ which holds for any $m: G \to H$ by Proposition~\ref{prop:image_R_independent}.

\begin{prop}
\label{prop:conj_group_morphism}
Consider some $\widetilde{F}$ obeying to Eq.~\eqref{eq:ftilde_commut}.
$\overline{F}: \Sym_\V \to \Sym_\V$ is a group homomorphism.
\end{prop}
\begin{proof}
By definition, $\overline{F}$ behaves exactly as $\widetilde{F}$ on the endomorphisms of $\emptygraph$.
Since $\widetilde{F}$ is functorial, $\overline{F}$ is an endofunctor on the group $\Sym_\V$, so a group homomorphism.
\end{proof}

The previous results can be summarized in the following commuting diagram.

\begin{center}
\begin{tikzpicture}
\node (GI) at (0,1.4) {$\Gsdp$};
\node (CI) at (2.2,1.4) {$\Csdp$};
\node (SI) at (4.4,1.4) {$\Sym_\V$};
\node (GO) at (0,0) {$\Gsdp$};
\node (CO) at (2.2,0) {$\Csdp$};
\node (SO) at (4.4,0) {$\Sym_\V$};

\draw[->] (GI) to node[above,style={scale=.85}]{$\U$} (CI);
\draw[->] (CI) to node[above,style={scale=.85}]{$|{-}|$} (SI);
\draw[->] (GO) to node[above,style={scale=.85}]{$\U$} (CO);
\draw[->] (CO) to node[above,style={scale=.85}]{$|{-}|$} (SO);
\draw[->] (GI) to node[left,style={scale=.85}]{$F$} (GO);
\draw[->] (CI) to node[left,style={scale=.85}]{$\widetilde{F}$} (CO);
\draw[->] (SI) to node[right,style={scale=.85}]{$\overline{F}$} (SO);
\end{tikzpicture}
\end{center}

The following proposition states how $\overline{F}$ is related to the CGDs renaming-invariance we want to capture.

\begin{prop}
\label{prop:ftilde_image_in_conj}
Consider some $\widetilde{F}$ obeying to Eq.~\eqref{eq:ftilde_commut}.
For any renaming $R$, $\overline{F}(R)$ is a conjugate of $R$, \ie, for any $G$, we have
$$
F(R(G)) = \overline{F}(R)(F(G)).
$$
\end{prop}

\begin{proof}
Consider the isomorphism $G_R: G \to R(G)$ with $|G_R| = R$.
By functoriality of $\widetilde{F}$, $\widetilde{F}(G_R): F(G) \to F(R(G))$ is an isomorphism as well.
By Proposition~\ref{prop:forms_of_iso}, $F(R(G)) = |\widetilde{F}(G_R)|(F(G)) = \overline{F}(|G_R|)(F(G)) = \overline{F}(R)(F(G))$ as required.
\end{proof}

As a summary, defining $\widetilde{F}$ consists in choosing, for each renaming $R$, a conjugate $\overline{F}(R) \in \Conj_F(R)$ in a functorial way.
Proposition~\ref{prop:cgd-rename-invariance} of CGDs renaming-invariance states that each $\Conj_F(R)$ is never empty.
However, they might not be singletons, in which case the challenge is to make a coherent choice to obtain a functor.
We put this challenge aside and place ourselves in the simpler case where all conjugates are unique (\ie, $\Conj_F(R)$ is a singleton for any $R$) which we call \emph{unique conjugate assumption}.
We will discuss later the general case.

\begin{prop}
Under the unique conjugate assumption, Eq.~\eqref{eq:ftilde_commut} characterizes uniquely the functor $\widetilde{F}: \Csdp \to \Csdp$.
\end{prop}

\begin{proof}
By Proposition~\ref{prop:image_R_independent}, the functor $\widetilde{F}$ is entirely defined by the behavior of $\overline{F}$.
By Proposition~\ref{prop:ftilde_image_in_conj}, we have $\overline{F}(R) \in \Conj_F(R)$ which is a singleton under the unique conjugate assumption.
So, $\widetilde{F}$ is all fixed in that case and it remains to show that it is a functor.
But this means that $\overline{F}: \Sym_\V \to \Sym_\V$ is functorial, which is exactly stated by Proposition~\ref{prop:conj_group_morphism}.
\end{proof}

\subsection{The Category of Disks}
\label{sec:cat-disks}

Like it has been done for graphs, disks which are related by renaming need to be considered as isomorphic.
This leads to consider a category from the poset $\Dsdp^r$.
This category inherits from $\Dsdp^r$ the fact that disks are compared by aligning their centers.
Consequently, the comma category $\widetilde{i}/G$ to be considered later (whatever $\widetilde{i}$ will be) will not be able to identify how two disks of $G$ with different centers can be related to each other.
However, this information is of main importance when computing the colimit.
Indeed, in the original poset setting, images of non aligned disks are well positioned thanks to the absolute positioning.
But this is lost when quotienting by renamings.
In a more general categorical setting, alignment has to be specified through morphisms directly.
This is a usual issue when dealing with GTs: rules have to specify not only the images of the disks but also how these images must be connected in the output.

We need to complete the quotienting with more information about disks interactions.
Past experiences in GTs show that there exist two main strategies to keep track of disks relationships: completing the category of disks either with unions of disks or intersections of disks~\cite{DBLP:conf/uc/FernandezMS19}.
We used the former in the case of CA~\cite{DBLP:journals/nc/FernandezMS23} where a category of \emph{pairs of disks} is considered.
In the present work, we focus on the latter strategy since the empty graph provides us with a common ``subdisk'' able to relate all proper disks.

\begin{defi}[Category of Disks]
\label{def:cat-diskpair}
Let $\CDsdp^r$ be the category with as objects the set
$$
\{ (H, \{c\}) \mid (H, c) \in \Dsdp^r \} \cup \{ (\emptygraph, \emptyset) \},
$$
and, for any two such objects $(H_1,C_1)$ and $(H_2,C_2)$, a morphism $m: (H_1,C_1) \to (H_2,C_2)$ is the data of a graph morphism (also abusively denoted) $m: H_1 \to H_2 \in \Csdp$ such that $|m|(C_1) \subseteq C_2$.
Composition is inherited from $\Csdp$, and $\id_{(H,C)} = \id_H$, for any object $(H,C)$.
\end{defi}

This category contains as objects all disks of $\Dsdp^r$ with an additional ``not centered empty disk''.
We use a set (singleton or empty set) to deal with the presence / absence of center and get a homogeneous representation.
Thanks to this encoding, a morphism from the empty disk does not require any constraint on the center ($|m|(\emptyset) = \emptyset \subseteq C_2$ whatever $C_2$ is), while a morphism between two proper disks requires them to be aligned by centers ($|m|(\{c_1\}) \subset \{c_2\} \iff |m|(c_1) = c_2$).

\begin{prop}
$\CDsdp^r$ is indeed a category.
\end{prop}

\begin{proof}
For any $(H,C) \in \CDsdp^r$, $\id_{(H,C)} = \id_H$ is well defined since we have $|\id_H|(C) = \id_\V(C) = C \subseteq C$.
The composition is well defined as well.
Indeed, given $m: (H,C) \to (H',C')$ and $n: (H',C') \to (H'',C'')$, we have
$$
|m|(C) \subseteq C' \land |n|(C') \subseteq C'' \implies |n \circ m|(C) = |n|(|m|(C)) \subseteq C''
$$
which implies that $n \circ m$ is a morphism of $\CDsdp^r$.
Associativity of the composition and neutrality of the identities are inherited from $\Csdp$ and already proved in Proposition~\ref{prop:graph_cat}.
\end{proof}

Notice that, under the common condition that the number of labels in $\Sigma$ and $\Delta$ is finite, this category of disks has now essentially finitely many disks.
Indeed, because of the bounds imposed by the radius and the number of ports, there is a maximal number $n$ of vertices.
If we fix a subset $W \subset \V$ of size $n$ and restrict to the disks $(H, c)$ such that $V(H) \subseteq W$, we have only finitely many disks as soon as there are finite labels.
Any other disk is isomorphic to one of those disks.

\subsection{The Pointer Dropping and Local Rule Functors.}

The next step is straightforward definitions of the categorical counterpart of functions $i$ and $f$.
On the one hand, the functor $\widetilde{i}$ drops the centers of disks.

\begin{defi}[Pointer Dropping Functor]
\label{def:funcrtor-itilde}
Let $\widetilde{i}: \CDsdp^r \to \Csdp$ be the functor defined by $\widetilde{i}((H,C)) = H$ and $\widetilde{i}(m: (H,C) \to (H',C')) = m$.
\end{defi}

\begin{prop}
$\widetilde{i}$ is indeed a functor.
\end{prop}

\begin{proof}
The proposition is trivial since $\widetilde{i}$ simply drops centers.
\end{proof}

Let now exhibit the structure induced on the comma category $\widetilde{i}/G$.
Formally, an object of $\widetilde{i}/G$ is a couple $\langle (H, C), m: H \to G\rangle$ composed of a disk $(H,C) \in \CDsdp^r$ together with a morphism $m: H \to G$ expressing an occurrence of that disk in $G$.
To avoid redundancies in notations, we simply write $\langle C, m: H \to G\rangle$.
Morphisms of $\widetilde{i}/G$ are couples $\langle n: (H_1, C_1) \to (H_2, C_2), m: H_2 \to G \rangle: \langle C_1, m \circ n \rangle \to \langle C_2, m \rangle$.
Notice that once given the specification of the couple, the signature of the morphism is completely known.
In the following, signatures of morphisms are then omitted.

\begin{prop}
\label{prop:comma-nonempty}
\label{prop:comma-connected}
For any $G \in \Csdp$, the comma category $\widetilde{i}/G$ is non-empty and connected.
\end{prop}
\begin{proof}
For non-emptiness, we consider the empty disk $(\emptygraph, \emptyset)$ which occurs in any graph $G$.
So we have $\langle \emptyset, \U(\emptygraph \subseteq G) \rangle$ is in $\widetilde{i}/G$.

To show the connectedness of $\widetilde{i}/G$, \ie, the fact the any two objects are connected by a non-oriented path of morphisms, consider two arbitrary disks $(H_1, C_1)$ and $(H_2, C_2)$ occurring in $G$ by $\langle C_1, m_1: H_1 \to G \rangle$ and $\langle C_2, m_2: H_2 \to G\rangle$.
Both are related to $\langle \emptyset, \U(\emptygraph \subseteq G) \rangle$ by the morphisms $\langle n_k: (\emptygraph,\emptyset) \to (H_k,C_k), m_k \rangle$ with $|n_k| = |m_k|^{-1}$ for $k \in \{ 1, 2 \}$.
Indeed, firstly $n_k$ is well defined since $|n_k|(\emptyset) = \emptyset \subseteq C_k$ and $|n_k|(\emptygraph) = \emptygraph \subseteq H_k$.
Secondly, $\U(\emptyset \subseteq G) = m_k \circ n_k$ since $|m_k \circ n_k| = |m_k| \circ |n_k| = |m_k| \circ |m_k|^{-1} = \id_\V$.
So we have a path of length 2 in $\widetilde{i}/G$ between any two arbitrary objects $\langle C_1, m_1\rangle$ and $\langle C_2, m_2\rangle$.
\end{proof}

On the other hand, the local rule functor $\widetilde{f}$ applies $f$ on disks and does nothing for the empty disk.

\begin{defi}[Local Rule Functor]
\label{def:funcrtor-ftilde}
Let $\widetilde{f}: \CDsdp^r \to \Csdp$ be the functor defined by $\widetilde{f}((H,\{c\})) = f((H,c))$ and $\widetilde{f}((\emptygraph,\emptyset)) = \emptygraph$, and for any $m: (H,C) \to (H',C')$, $|\widetilde{f}(m)| = \overline{F}(|m|)$.
\end{defi}

\begin{prop}
$\widetilde{f}$ is indeed a functor.
\end{prop}

\begin{proof}
We first need to show that $\widetilde{f}$ is well defined, \ie, that $\widetilde{f}(m)$ is a morphism from $\widetilde{f}((H,C))$ to $\widetilde{f}((H',C'))$.
So we expect $|\widetilde{f}(m)|(\widetilde{f}((H,C))) \subseteq \widetilde{f}((H',C'))$, which rewrites into $|\overline{F}(|m|)(\widetilde{f}((H,C))) \subseteq \widetilde{f}((H',C'))$ by definition of $\widetilde{f}$.
Let us look at the three possible cases.
\begin{enumerate}

    \item If $(H,C) = (H',C') = (\emptygraph, \emptyset)$, $\widetilde{f}((H,C)) = \widetilde{f}((H',C')) = \emptygraph$ and $\overline{F}(|m|)(\emptygraph) = \emptygraph \subseteq \emptygraph$.

    \item If $(H,C) = (\emptygraph, \emptyset)$ and $(H',C') = (H',\{c'\})$, $\widetilde{f}((H,C)) = \emptygraph$ and $\widetilde{f}((H',C')) = f((H',c'))$;
so, $\overline{F}(|m|)(\emptygraph) = \emptygraph \subseteq f((H',c'))$.

    \item If $(C,H) = (H, \{c\})$ and $(H',C') = (H', \{c'\})$, $\widetilde{f}((H,C)) = f((H,c))$ and $\widetilde{f}((H',C')) = f((H',c'))$;
so, we need to check that $\overline{F}(|m|)(f((H, c))) \subseteq f((H', c'))$.
Observe that if $\Conj_f(|m|) = \{ \overline{F}(|m|) \}$, we get the result: since $|m|(c) = c'$ and $|m|(H) \subseteq H'$, we have $\overline{F}(|m|)(f((H, c))) = f((|m|(H), |m|(c))) \subseteq f((H',c'))$ by monotony of $f$.
We now show that $\Conj_f(|m|) = \{ \overline{F}(|m|) \}$.
By Definition~\ref{def:graph_local_rule} of local rules, $f$ has a conjugate for $|m|$.
So, by Proposition~\ref{prop:cgd-rename-invariance}, we get $\emptyset \subsetneq \Conj_f(R) \subseteq \Conj_F(R)$ for any renaming $R$.
By unique conjugate assumption, $f$ has the same conjugate as $F$ for $|m|$.
So $\Conj_f(|m|) = \{ \overline{F}(|m|) \}$.
\end{enumerate}
Finally, functoriality is obtained by functoriality of $\overline{F}$ (Proposition~\ref{prop:conj_group_morphism}).
\end{proof}

Notice that the pointer dropping functor and the local rule functor are completely determined by their restriction to the subcategory of disks with finitely many objects described at the end of Section~\ref{sec:cat-disks}.
This can be seen by rewriting the defining equation of $\widetilde{f}$ given in Def.~\ref{def:funcrtor-ftilde}. Indeed, for objects, one has the equalities $\widetilde{f}((H,\{c\})) = f((H,c)) = f(R(H',c')) = R'(f((H',c'))) = |\widetilde{f}(R:(\emptygraph,\emptyset) \to (\emptygraph,\emptyset))|(\widetilde{f}((H',c')))$ for some $(H',\{c'\})$ in the subcategory, $R$ sending $(H',\{c'\})$ to $(H,\{c\})$, and any $R' \in \Conj_f(R)$. For morphisms, note that the proof of Proposition~\ref{prop:image_R_independent} also holds for the local rule functor.

\subsection{CGDs with Unique Conjugate as Kan Extensions}

We end our journey with the main result of the section stating that $\widetilde{F}$ is the pointwise left Kan extension of $\widetilde{f}$ along $\widetilde{i}$.
To get the result, we first focus on an input $G \in \Csdp$ and a colimit over $\widetilde{f}\circ\Proj_{\widetilde{i}/G}$, \ie a universal cocone from $\widetilde{f}\circ\Proj_{\widetilde{i}/G}$ to $F(G)$.

\begin{defi}
\label{def:cocone-theta}
Let $\theta^G: \widetilde{f}\circ\Proj_{\widetilde{i}/G} \Rightarrow F(G)$ be the cocone with components
$$
\theta^G_{\langle C, m: H \to G \rangle}: \widetilde{f}((C,H)) \to F(G)
$$
such that $|\theta^G_{\langle C, m: H \to G \rangle}| := \overline{F}(|m|)$.
\end{defi}

\begin{prop}
$\theta^G$ is indeed a cocone.
\end{prop}
\begin{proof}
We need to show that the components are well-defined morphisms, then that they commute with the diagram morphisms.
Consider some $\langle C, m: H \to G \rangle$ of $\widetilde{i}/G$.
We want to show that $\theta^G_{\langle C, m\rangle}$ is well defined, which means that $|\theta^G_{\langle C, m\rangle}|(\widetilde{f}((C,H))) \subseteq F(G)$, that is, $\overline{F}(|m|)(\widetilde{f}((C,H))) \subseteq F(G)$ by Definition~\ref{def:cocone-theta}.
Let us proceed by case.
For $(H,C)=(\emptygraph, \emptyset)$, it rewrites to $\overline{F}(|m|)(\emptygraph) = \emptygraph \subseteq F(G)$ which holds trivially.
For $(H,C)=(H,\{c\})$, it rewrites into $\overline{F}(|m|)(f((H,c))) \subseteq F(G)$.
Since $\overline{F}(|m|) \in \Conj_f(|m|)$, we get $\overline{F}(|m|)(f((H,c))) = f((|m|(H),|m|(c)))$.
Since $m: H \to G$ and $(H,c) \in \Dsdp^r$, $(|m|(H),|m|(c)) \subseteq G^r_c$.
By monotonicity of $f$ then by Eq.~\eqref{eq:cgd}, $f((|m|(H),|m|(c))) \subseteq f(G^r_v) \subseteq F(G)$ as required.

For commutation, take some morphism $\langle n: (H_1, C_1) \to (H_2, C_2), m: H_2 \to G \rangle$ in $\widetilde{i}/G$.
We have to check $\theta^G_{\langle C_1, m \circ n \rangle} = \theta^G_{\langle C_2, m \rangle} \circ \widetilde{f}(n)$.
Since the signature is correct, it remains to check the equality of the associated renaming.
\begin{align*}
|\theta^G_{\langle C_1, m \circ n \rangle}|
& = \overline{F}(|m \circ n|) & \text{by Def.~\ref{def:cocone-theta}} \\
& = \overline{F}(|m|) \circ \overline{F}(|n|) & \text{by Prop.~\ref{prop:conj_group_morphism} and Def.~\ref{def:graph_cat}} \\
& = |\theta^G_{\langle C_2, m \rangle}| \circ |\widetilde{f}(n)| & \text{by Def.~\ref{def:cocone-theta} and Def.~\ref{def:funcrtor-ftilde}} \\
& = |\theta^G_{\langle C_2, m \rangle} \circ \widetilde{f}(n)| & \text{by Def.~\ref{def:graph_cat}}
& \qedhere
\end{align*}
\end{proof}

\begin{prop}
\label{prop:theta-universal-cocone}
$\theta^G$ is a universal cocone.
\end{prop}
\begin{proof}
Suppose a cocone $\xi: \widetilde{f} \circ \Proj_{\widetilde{i}/G} \Rightarrow G'$.
We have to show the existence of a unique morphism $\lambda: F(G) \to G'$ such that for all $\langle C, m: H \to G \rangle$, $\xi_{\langle C, m\rangle} = \lambda \circ \theta^G_{\langle C, m \rangle}$.
If such a morphism exists, it is uniquely defined by its associated renaming $|\lambda|$.
Each $\langle C, m \rangle$ for $\widetilde{i}/G$ comes with the constraint that $|\lambda| = |\xi_{\langle C, m \rangle}| \circ |\theta^G_{\langle C, m\rangle}|^{-1}$.
Let us set $\lambda^{\langle C, m\rangle} := |\xi_{\langle C, m \rangle}| \circ |\theta^G_{\langle C, m\rangle}|^{-1}$, and show that $\lambda^{\langle C_1, m_1\rangle} = \lambda^{\langle C_2, m_2\rangle}$ for any $\langle C_1, m_1: H_1 \to G \rangle$ and $\langle C_2, m_2: H_2 \to G \rangle$.
To do so, since $\widetilde{i}/G$ is connected by Proposition~\ref{prop:comma-connected}, it is enough to notice that it holds when $\langle C_1, m_1 \rangle$ and $\langle C_2, m_2 \rangle$ are related by a morphism in $\widetilde{i}/G$.
Suppose a morphism $\langle n: (H_1, C_1) \to (H_2, C_2), m: H_2 \to G \rangle$.
We have $\lambda^{\langle C_1, m \circ n \rangle} = \lambda^{\langle C_2, m \rangle}$.
Indeed, notice that since $\xi$ and $\theta^G$ are cocones, we have $\xi_{\langle C_1, m \circ n \rangle} = \xi_{\langle C_2, m \rangle} \circ \widetilde{f}(n)$ and $\theta^G_{\langle C_1, m \circ n \rangle} = \theta^G_{\langle C_2, m \rangle} \circ \widetilde{f}(n)$.
So,
\begin{align*}
\lambda^{\langle C_1, m \circ n \rangle}
& = |\xi_{\langle C_1, m \circ n \rangle}| \circ |\theta^G_{\langle C_1, m \circ n \rangle}|^{-1} \\
& = |\xi_{\langle C_2, m \rangle} \circ \widetilde{f}(n)| \circ |\theta^G_{\langle C_2, m \rangle} \circ \widetilde{f}(n)|^{-1} \\
& = |\xi_{\langle C_2, m \rangle}| \circ |\widetilde{f}(n)| \circ |\widetilde{f}(n)|^{-1} \circ |\theta^G_{\langle C_2, m \rangle}|^{-1} \\
& = |\xi_{\langle C_2, m \rangle}| \circ |\theta^G_{\langle C_2, m \rangle}|^{-1} \\
& = \lambda^{\langle C_2, m \rangle}
\end{align*}
as expected.
Since $\widetilde{i}/G$ is not empty by Proposition~\ref{prop:comma-nonempty} and all $\lambda^{\langle C, m \rangle}$ agree, we are able to consider $\lambda$ such that $|\lambda| = \lambda^{\langle C, m \rangle}$.
It remains to show that $\lambda$ defined with such an associated renaming is a morphism from $F(G)$ to $G'$.
Equivalently, we have to show that $|\lambda|(F(G)) \subseteq G'$.
\begin{align*}
|\lambda|(F(G))
& = |\lambda|(\bigcup_{c \in V(G)} f(G^r_c)) \\
& = \bigcup_{c \in V(G)}|\lambda|(f(G^r_c)) \\
& = \bigcup_{c \in V(G)} \lambda^{\langle \{c\}, \U(i(G^r_c) \subseteq G) \rangle}(f(G^r_c)) \\
& = \bigcup_{c \in V(G)} |\xi_{\langle \{c\}, \U(i(G^r_c) \subseteq G) \rangle}|(|\theta^G_{\langle \{c\}, \U(i(G^r_v) \subseteq G) \rangle}|^{-1}(f(G^r_c))) \\
& = \bigcup_{c \in V(G)} |\xi_{\langle \{c\}, \U(i(G^r_c) \subseteq G) \rangle}|(\overline{F}(\id_\V)^{-1}(f(G^r_c))) \\
& = \bigcup_{c \in V(G)} |\xi_{\langle \{c\}, \U(i(G^r_c) \subseteq G) \rangle}|(f(G^r_c)) \\
& \subseteq \bigcup_{c \in V(G)} G' \\
& = G'
\end{align*}
The inclusion comes from the fact that $\xi_{\langle \{c\}, \U(i(G^r_c) \subseteq G) \rangle}$ is a morphism from $f(G^r_c)$ to $G'$.
This concludes the proof of universality of $\theta^G$.
\end{proof}

\begin{prop}
$\widetilde{F}$ is the pointwise left Kan extension of $\widetilde{f}$ along $\widetilde{i}$.
\end{prop}
\begin{proof}
For any graph $G$, since $\theta^G$ is a universal cocone with apex $F(G)$ by Proposition~\ref{prop:theta-universal-cocone}, we already have that $\widetilde{F}(G) \cong \Colim \widetilde{f} \circ \Proj_{\widetilde{i}/G}$.
It remains to show the functoriality of the construction which means that for any $h: G \to G'$, $\widetilde{F}(h) \cong \Colim \widetilde{f} \circ \Proj_{\widetilde{i}/h}$.
On the right hand side, the expressed morphism is the unique mediating $\lambda$ from $\theta^G$ to $\theta^{G'}_{\widetilde{i}/h}$.
To clarify, ${\widetilde{i}/h}$ is the functor embedding $\widetilde{i}/G$ into $\widetilde{i}/G'$ mapping any $\langle n: (H_1, C_1) \to (H_2, C_2), m: H_2 \to G \rangle$ to $\langle n, h \circ m \rangle$.
So, $\theta^{G'}_{\widetilde{i}/h}$ is the cocone over $\widetilde{f} \circ \Proj_{\widetilde{i}/G}$ with components $\theta^{G'}_{(\widetilde{i}/h)(\langle C, m: H \to G \rangle)} = \theta^{G'}_{\langle C, h \circ m \rangle}$.
By universality of $\theta^G$ (Proposition~\ref{prop:theta-universal-cocone}), we get a unique morphism $\lambda: F(G) \to F(G')$.
Using the elements of the proof of Proposition~\ref{prop:theta-universal-cocone}, we use the object $\langle \emptyset, \U(\emptygraph \subseteq G) \rangle \in \widetilde{i}/G$ to get
\begin{align*}
|\lambda|
& = \lambda^{\langle \emptyset, \U(\emptygraph \subseteq G) \rangle} \\
& = |\theta^{G'}_{(\langle \emptyset, h \circ \U(\emptygraph \subseteq G) \rangle)}| \circ |\theta^{G}_{(\langle \emptyset, \U(\emptygraph \subseteq G) \rangle)}|^{-1} \\
& = \overline{F}(|h \circ \U(\emptygraph \subseteq G) \rangle)|) \circ \overline{F}(|\U(\emptygraph \subseteq G) \rangle)|)^{-1} \\
& = \overline{F}(|h|) \circ \overline{F}(|\U(\emptygraph \subseteq G) \rangle)|) \circ \overline{F}(|\U(\emptygraph \subseteq G) \rangle)|)^{-1} \\
& = \overline{F}(|h|) \\
& = |\widetilde{F}(h)|.
\end{align*}
We conclude that $\lambda = \widetilde{F}(h)$ (both have same signature and underlying renaming) which proves the required functoriality.
\end{proof}

\section{Conclusion}
\label{sec:concl}

In this article, we planned to compare CGD and GT frameworks.
The very particular route we have chosen for this task led us to identify the class of monotonic CGDs (for a particular partial order on graphs) which are both CGDs and GTs, and happen to be universal among all CGDs.
We have then transformed the renaming-invariance property of the dynamics into additional relations between the graphs, so that the resulting notion of monotonicity (functoriality in fact) now ingrates renaming-invariance.
There are at least two benefits for that.
Firstly this is more in line with usual categories for graphs where no notion of absolute names or positioning exists, but only their (possibly multiple) occurrences in each other.
Secondly, this new formulation permits, in the particular case where there is a finite number of labels, to make explicit the finite nature of the local rule by expressing directly that there are finitely many disks, but that each of them can possibly be found at many places in a given graph.

\paragraph{There is still more work to do to finish this line of reasoning though.}

Firstly, we have only treated the case of CGDs with ``unique conjugates''.
One may think that it is a strong assumption, but as soon as a single graph $G$ has an output $F(G)$ without symmetry, the functor $F$ has ``unique \emph{non-spurious} conjugates'' (as is the case for the moving particle CGD example used in Section~\ref{sec:universal-mono-cgd} given the asymmetric outputs in Figure~\ref{fig:monotonic-particle}).
To see this, note that all conjugates are determined by the conjugates of the identity renaming, since $\Conj_F(R) = R' \circ \Conj_F(\id_\V)$ for any $R' \in \Conj_F(R)$.
Conjugates $S \in \Conj_F(\id_\V)$ of the identity renaming are closely related to symmetries via their defining property $S(F(G)) = F(G)$ for all $G$.
Notice that, for a given $G$, two $S, S' \in \Conj_F(\id_\V)$ may act similarly on $V(F(G))$ while acting distinctly on $\mathcal{V} \setminus V(F(G))$; in other words, many conjugates may express the same symmetry of $F(G)$.
Going further, conjugates can arbitrarily deal with names of $\mathcal{V}' = \mathcal{V} \setminus \bigcup_G V(F(G))$ giving a spurious impression of a large amount of meaningful conjugates.
It is in fact enough to restrict the use of conjugates to \emph{non-spurious} conjugates, that is, any conjugate acting as the identity on $\mathcal{V'}$.
For the cases where there is at least one $G$ with an asymmetric $F(G)$, all the results of the paper hold using the unique \emph{non-spurious} conjugate.
For the other cases, the non-spurious conjugates of the identity invite to consider the quotient of the outputs of $F$ with respect to those symmetries, hopefully yielding a ``unique non-spurious conjugates'' dynamics and making those dynamics somehow universal.

Secondly, while the category of disks has essentially finitely many objects (as discussed at the end of Section~\ref{sec:cat-disks}), there are still infinitely many morphisms between them.
This is because the chosen notion of morphisms is based on bijections of the infinite set $\V$, in order to have a first categorical formulation as close as possible to the framework of CGDs.
It should be possible to get rid of $\V$ completely by introducing in the definition of a graph the arbitrary set of vertices it uses (instead of a subset of $\V$).
In that case, a morphism would be an injective homomorphism of graph, and there would be a finite number of them when the graphs are finite.
In other words, it should be possible to work in a more classical categorical fashion.
In fact, the structure of CGDs port graphs is reminiscent of an important part of the literature concerning graph rewriting.
In this line, CGDs (and the categorical interpretation proposed in this paper) have to be compared with transformations of multigraphs with ports of~\cite{andrei2008rewriting} which have numerous applications in system modeling, and with site graphs of~\cite{danos:hal-00809065,danos:hal-01976370} invoked in the formal description of the Kappa modeling language.
Both of these contributions originally work on finite graphs and rely on the local application of a rewriting rule.
Using the GT framework on such formal graphs should converge to a proposition similar to what has been done in this paper where the global set of names is dropped away in favor of local sets of names.
In this perspective, the present work represents a milestone that can be used as a reference for unraveling the relationship between theses approaches at the formal level.

\paragraph{Coming back to Section~\ref{sec:cgd_post_lan}.}

This work was guided by the formal similarities between the two frameworks leading to a list of four questions.
The chosen strategy to cope with these questions was to tackle the first one ``What is the order?'' with the goal of answering positively to the second question ``Is the union of CGDs the supremum of this order?''.
This journey led us to identify the subgraph order to structure the set of port graphs instead of considering them independent (more precisely related by the disks only) as the original framework does.
This pushes forward the idea of using graph inclusion to express gain of information as proposed by the GT framework, opening a new direction when designing CGDs by respecting the order with monotonicity.

One interesting thing to note is that a slight adaptation of the original definition of port graphs allows to represent general CGDs as Kan extensions without any encoding.
Indeed, the encoding considered in this article leads one to consider a collection of graphs where most of them are ill-formed.
By adding explicitly to graphs additional features to represent the positive information that some port is not occupied or some label is missing, it is possible to get an encoding where all objects make sense.
This is related to how the notion of \emph{stub} on site graphs presented in~\cite{danos:hal-01976370} explicitly specifies positive information that some site is unbound.

An alternative route to answer the four-question list is possible by tackling the first question with the goal of answering the third one (almost) positively, thus falsifying the second one.
An order closely related to this route is the ``induced subgraph'' order stating that $G$ is lower than $H$ if $G = H^0_{V(G)}$.
This order is stronger than the usual subgraph order but is also very interesting.
Fewer graphs are ``inducedly consistent'' and one may ask if the subclass of CGDs definable with this stronger notion of consistency is universal.
This might give an idea on whether the result of this article is isolated or, on the contrary, if it follows a common pattern shared with many instances.

\bibliographystyle{alphaurl}
\bibliography{ICGT24-extver-submit4}

\end{document}